\documentclass[11pt,reqno]{amsart}

\usepackage{amssymb,amsmath,amsthm,amsfonts}
\usepackage{subfigure, hyperref, multirow, bm, color, stmaryrd, marginnote, graphicx, wrapfig}

\usepackage[margin=1in]{geometry}
\usepackage{etoolbox}
\patchcmd{\section}{\scshape}{\bfseries\Large}{}{}
\makeatletter
\renewcommand{\@secnumfont}{\bfseries}
\makeatother
\makeatletter
\renewcommand{\tocsection}[3]{%
  \indentlabel{\@ifnotempty{#2}{\bfseries\ignorespaces#1 #2\quad}}\bfseries#3}
\renewcommand{\tocsubsection}[3]{%
  \indentlabel{\@ifnotempty{#2}{\ignorespaces#1 #2\quad}}#3}

\newcommand\@dotsep{4.5}
\def\@tocline#1#2#3#4#5#6#7{\relax
  \ifnum #1>\c@tocdepth 
  \else
    \par \addpenalty\@secpenalty\addvspace{#2}%
    \begingroup \hyphenpenalty\@M
    \@ifempty{#4}{%
      \@tempdima\csname r@tocindent\number#1\endcsname\relax
    }{%
      \@tempdima#4\relax
    }%
    \parindent\z@ \leftskip#3\relax \advance\leftskip\@tempdima\relax
    \rightskip\@pnumwidth plus1em \parfillskip-\@pnumwidth
    #5\leavevmode\hskip-\@tempdima{#6}\nobreak
    \leaders\hbox{$\m@th\mkern \@dotsep mu\hbox{.}\mkern \@dotsep mu$}\hfill
    \nobreak
    \hbox to\@pnumwidth{\@tocpagenum{\ifnum#1=1\bfseries\fi#7}}\par
    \nobreak
    \endgroup
  \fi}
\AtBeginDocument{%
\expandafter\renewcommand\csname r@tocindent0\endcsname{0pt}
}
\def\l@subsection{\@tocline{2}{0pt}{2.5pc}{5pc}{}}
\makeatother

\newcommand{\R}{\mathbb{R}}
\newcommand{\Z}{\mathbb{Z}}

\newcommand{\T}{\mathbb{T}}

\newcommand{\E}{\mathcal{E}}

\newcommand{\bars}{\overline s}

\newcommand{\bu}{\bm{u}}

\newcommand{\bx}{\bm{x}}

\newcommand{\X}{\bm{X}}
\newcommand{\be}{\bm{e}}

\newcommand{\bv}{\bm{v}}
\newcommand{\p}{\partial}
\renewcommand{\div}{{\rm{div}\,}}
\newcommand{\SB}{{\rm SB}}

\newcommand{\abs}[1]{\left\lvert #1 \right\rvert}
\newcommand{\norm}[1]{\left\lVert #1 \right\rVert}
\newcommand{\wh}[1]{\widehat{#1}}

\newcommand{\mc}[1]{\mathcal{#1}}

\newtheorem{theorem}{Theorem}[section]
\newtheorem{lemma}[theorem]{Lemma}

\theoremstyle{definition}
\newtheorem{remark}[theorem]{Remark}

\allowdisplaybreaks

\begin{document}
\title{Remarks on regularized Stokeslets in slender body theory}

\author{Laurel Ohm}
\address{Courant Institute of Mathematical Sciences, New York University, New York, NY 10012}
\email{lo2083@cims.nyu.edu}



\begin{abstract} 
We remark on the use of regularized Stokeslets in the slender body theory (SBT) approximation of Stokes flow about a thin fiber of radius $\epsilon>0$. Denoting the regularization parameter by $\delta$, we consider regularized SBT based on the most common regularized Stokeslet plus a regularized doublet correction. Given sufficiently smooth force data along the filament, we derive $L^\infty$ bounds for the difference between regularized SBT and its classical counterpart in terms of $\delta$, $\epsilon$, and the force data. We show that the regularized and classical expressions for the velocity of the filament itself differ by a term proportional to $\log(\delta/\epsilon)$ -- in particular, $\delta=\epsilon$ is necessary to avoid an $O(1)$ discrepancy between the theories. However, the flow at the surface of the fiber differs by an expression proportional to $\log(1+\delta^2/\epsilon^2)$, and any choice of $\delta\propto\epsilon$ will result in an $O(1)$ discrepancy as $\epsilon\to 0$. Consequently, the flow around a slender fiber due to regularized SBT does not converge to the solution of the well-posed \emph{slender body PDE} which classical SBT is known to approximate. Numerics verify this $O(1)$ discrepancy but also indicate that the difference may have little impact in practice. 
\end{abstract}

\maketitle

\section{Introduction}
The method of regularized Stokeslets was introduced by Cortez in \cite{cortez2001method} to eliminate the need to integrate a singular kernel in boundary integral methods for Stokes flow. Since then, regularized Stokeslets have enjoyed widespread use in models of general fluid--structure interaction \cite{ainley2008method, cogan2008regularized,cortez2005method,cortez2015general,fauci2006biofluidmechanics,nguyen2014computing,sun2015boundary}.
The method of regularized Stokeslets has become especially popular for modeling the dynamics of thin fibers in a three dimensional fluid, providing an alternative way to deal with the singular integrals arising in the classical slender body theories of Lighthill \cite{lighthill1976flagellar}, Keller--Rubinow \cite{keller1976slender}, and Johnson \cite{johnson1980improved}. Slender body models based on regularized Stokeslets have found wide use in biophysical modeling applications \cite{bouzarth2011modeling2,bouzarth2011modeling1, buchmann2018mixing, buchmann2015flow, cisneros2010fluid, gillies2009hydrodynamic, lee2014nonlinear, olson2013modeling} and have led to advances in numerical methods \cite{gallagher2019sharp, gallagher2020passively, rostami2016kernel, smith2009boundary, smith2018nearest} as well as further modeling extensions \cite{cortez2018regularized, walker2019filament,walker2020regularised}. \\

In this paper we aim to compare slender body theory based on regularized Stokeslets \cite{cortez2012slender} to its classical counterpart \cite{gotz2000interactions,johnson1980improved,keller1976slender,lighthill1976flagellar,shelley2000stokesian,tornberg2004simulating}. This type of question has been addressed previously in various forms. Some authors have studied the convergence of a single regularized Stokeslet to a singular Stokeslet in the far field and as the regularization parameter $\delta\to 0$ \cite{mitchell2019there,nguyen2014reduction, zhao2019method}.
Others have considered regularized SBT specifically, as a line integral of either regularized Stokeslets \cite{bouzarth2011modeling1} or regularized Stokeslets plus doublet corrections \cite{cortez2012slender}, and have explored the relation to classical SBT in terms of both the fiber radius $\epsilon$ and the regularization parameter $\delta$. In particular, Cortez and Nicholas \cite{cortez2012slender} used a matched asymptotic expansion of the regularized SBT velocity field at the filament surface to compare the regularized fiber velocity to the classical Lighthill \cite{lighthill1976flagellar} and Keller--Rubinow \cite{keller1976slender} slender body theories. \\

Here we consider the regularized SBT expression used in Cortez--Nicholas \cite{cortez2012slender}, but rather than comparing an asymptotic expansion of regularized SBT to classical SBT, we compare the direct evaluation of the regularized expression along the fiber centerline, which is more like what is used in practice. We also consider the fluid flow around the slender body generated by regularized SBT and compare this to classical SBT. 
Recently, classical slender body theory has been shown to approximate the solution to a well-posed boundary value problem for Stokes flow around a slender 3D filament \cite{closed_loop,free_ends}, placing the classical approximation on rigorous theoretical footing. We aim to also consider regularized SBT within this framework. \\

We show the following results.
For the velocity of the fiber itself, while it is known that the regularization parameter $\delta$ should be chosen to be proportional to the fiber radius $\epsilon$ for highest accuracy, we show that in fact $\delta=\epsilon$ is necessary to achieve asymptotic agreement between regularized and classical SBT as $\epsilon\to 0$. 
However, the flow at the surface of the filament due to regularized SBT cannot be made to converge to the flow due to classical SBT for any choice of $\delta$ proportional to $\epsilon$. An $O(1)$ difference in the force-per-unit-length at the fiber surface is also noted.
As a consequence of this discrepancy, the flow around the filament due to regularized Stokeslets does not converge to the flow due to classical SBT as $\epsilon\to 0$. This discrepancy is perhaps not unexpected: since the leading order term of SBT blows up like $\log\epsilon$ as $\epsilon\to 0$, different methods designed to capture the leading order in $\epsilon$ may still differ by an $O(1)$ quantity. Part of the issue stems from balancing the effects of two small parameters, $\epsilon$ and $\delta$, and trying to relate them in a physically meaningful yet practically useful way. This type of issue does not arise for regularized Stokeslets in 2D or over surfaces in 3D. \\

Finally, we verify the above differences numerically and provide numerical evidence that the small magnitude and extent of the discrepancy between regularized and classical SBT may mean that their difference is more of a moral issue than a practical one.


\subsection{Classical slender body theory and the slender body PDE}\label{sec:geomANDsbt}
We begin with a precise definition of the slender body geometry considered here. We consider a closed, non-self-intersecting $C^2$ curve $\Gamma_0\in \R^3$ whose coordinates are given by $\X:\T = \R/\Z \to \R^3$, where $\X$ is parameterized by arclength $s$. The non-self-intersection requirement means that there exists a constant $c_\Gamma>0$ such that 
\begin{align*}
	\inf_{s\neq s'}\frac{\abs{\X(s)-\X(s')}}{\abs{s-s'}} \ge c_\Gamma.
\end{align*}

At each point along $\Gamma_0$ we define the unit tangent vector 
\begin{align*}
\be_{\rm s}(s) = \frac{d\X}{ds},
\end{align*}
and consider the orthonormal frame $(\be_{\rm s}(s),\be_{n_1}(s),\be_{n_2}(s))$ defined in \cite{closed_loop}, which satisfies the ODEs 
\begin{equation}\label{frameODE}
\frac{d}{ds} \begin{pmatrix}
\be_{\rm s}(s) \\
\be_{n_1}(s) \\
\be_{n_2}(s)
\end{pmatrix} =
\begin{pmatrix}
0 & \kappa_1(s) & \kappa_2(s) \\
-\kappa_1(s) & 0 & \kappa_3 \\
-\kappa_2(s) & \kappa_3 & 0
\end{pmatrix} \begin{pmatrix}
\be_{\rm s}(s) \\
\be_{n_1}(s) \\
\be_{n_2}(s)
\end{pmatrix}.
\end{equation}
Here $\kappa_3$, $\abs{\kappa_3}<\pi$, is a constant enforcing periodicity of the frame in $s$, and $\kappa_1(s)$ and $\kappa_2(s)$ are related to the curvature $\kappa(s)=|\frac{d^2\X}{ds}|$ of $\Gamma_0$ via
\begin{align*}
 \kappa(s) = \sqrt{\kappa_1^2(s) +\kappa_2^2(s)}.
 \end{align*} 
 We define 
 \begin{align*}
 \kappa_{\max}= \max_{s\in\T}\abs{\kappa(s)}.
 \end{align*}

We also define the following cylindrical unit vectors with respect to the orthonormal frame \eqref{frameODE}: 
\begin{equation}\label{erho}
\be_\rho(s,\theta) = \cos\theta\be_{n_1}(s) +\sin\theta\be_{n_2}(s), \quad
\be_\theta(s,\theta) = -\sin\theta\be_{n_1}(s) +\cos\theta\be_{n_2}(s). 
\end{equation}

We may then define the slender body $\Sigma_\epsilon$ and its surface $\Gamma_\epsilon$ as
\begin{align*}
\Sigma_\epsilon = \{ \bx\in\R^3 \; : \; \bx = \X(s) + \rho \be_\rho(s,\theta), \; s\in \T, \rho <\epsilon \}; \quad \Gamma_\epsilon = \p\Sigma_\epsilon = \X(s) +\epsilon \be_\rho(s,\theta).
\end{align*}

Along the slender body surface $\Gamma_\epsilon$ we define the Jacobian factor 
\begin{align}\label{Jeps}		
\mc{J}_\epsilon(s,\theta) = \abs{\frac{\p\Gamma_\epsilon}{\p s}\times \frac{\p\Gamma_\epsilon}{\p \theta}} &= \epsilon\big(1-\epsilon\wh\kappa(s,\theta) \big),\\
\wh\kappa(s,\theta) &:= \kappa_1(s)\cos\theta+\kappa_2(s)\sin\theta. \label{kappahat}
\end{align}		

We consider $\Sigma_\epsilon$ immersed in Stokes flow, where the fluid velocity $\bu(\bx)$ and pressure $p(\bx)$ satisfy 
\begin{align*}
-\mu\Delta \bu +\nabla p &= 0, \\
 \div \bu &= 0 \quad \text{in } \Omega_{\epsilon} = \R^3 \backslash \overline{\Sigma_{\epsilon}} .
\end{align*}
Here $\mu$ is the fluid viscosity, which for convenience we will henceforth scale to $\mu=1$. \\

The fundamental idea of slender body theory is to exploit the slenderness of the body to approximate the 3D fluid velocity about $\Sigma_\epsilon$ as the flow due to a 1D curve of point forces along $\Gamma_0$. The most basic form of SBT, due to figures such as Hancock \cite{hancock1953self}, Cox \cite{cox1970motion}, and Batchelor \cite{batchelor1970slender}, places a curve of \emph{Stokeslets}, the free-space Green's function for the Stokes equations, along the centerline of the filament. In $\R^3$, the Stokeslet is given by
\begin{equation}\label{stokeslet}
\mc{S}(\bx) = \frac{{\bf I}}{\abs{\bx}} + \frac{\bx\bx^{\rm T}}{|\bx|^3}.
\end{equation}

For improved accuracy, other authors, including Keller and Rubinow \cite{keller1976slender}, Lighthill \cite{lighthill1976flagellar}, and Johnson \cite{johnson1980improved}, include higher order corrections to the Stokeslet to account for the finite radius of the fiber. The natural choice of correction, often known as the \emph{doublet}, is given by the Laplacian of the Stokeslet: 
\begin{equation}\label{doublet}
\mc{D}(\bx) = \frac{1}{2}\Delta\mc{S}(\bx) =\frac{{\bf I}}{\abs{\bx}^3} - \frac{3\bx\bx^{\rm T}}{\abs{\bx}^5}.
\end{equation}

Combining the Stokeslet and doublet correction leads to the following approximation for the fluid velocity about the slender body: 
\begin{equation}\label{classicalSBT}
\bu^\SB(\bx) = \frac{1}{8\pi}\int_\T \bigg(\mc{S}(\bx-\X(s')) + \frac{\epsilon^2}{2}\mc{D}(\bx-\X(s')) \bigg) \bm{f}(s') \, ds'.
\end{equation}
The doublet coefficient $\frac{\epsilon^2}{2}$ is chosen to cancel (to leading order in $\epsilon$) the angular dependence of the velocity field \eqref{classicalSBT} at the fiber surface $\Gamma_\epsilon$, giving rise to an expression which to leading order depends only on arclength. Specifically, for $\bx=\X(s)+\epsilon\be_\rho(s,\theta)\in\Gamma_\epsilon$, expression \eqref{classicalSBT} satisfies \cite{closed_loop,free_ends}
\begin{equation}\label{fiber_integrity}
\abs{\frac{\p \bu^\SB}{\p\theta}(\bx)} \le C\epsilon\abs{\log\epsilon}\norm{\bm{f}}_{C^1(\T)}.
\end{equation}
This $\theta$-independence is sometimes known as the \emph{fiber integrity condition} \cite{closed_loop,free_ends} and ensures that cross sections of the slender body do not deform. In particular, the fiber integrity condition, enforced by including the doublet correction, serves to give the fiber a 3D structure that can be detected by the surrounding fluid velocity field.  
Throughout the paper we will refer to expression \eqref{classicalSBT} as \emph{classical} slender body theory. \\

One difficulty with the classical SBT expression \eqref{classicalSBT} is that the Stokeslet $\mc{S}$ and doublet $\mc{D}$ are singular at $\bx=\X(s)$, which raises the question of how to obtain a useful expression for the velocity of the fiber itself. \\

The most common method for obtaining a fiber velocity expression is to perform a matched asymptotic expansion of \eqref{classicalSBT} about $\epsilon=0$. The resulting expression, which we will denote by $\bu^\SB_{\rm{C}}(\X(s))$ to emphasize that the integral kernel is distinct from \eqref{classicalSBT}, is given by 
\begin{equation}\label{KR_SBT}
\begin{aligned}
8\pi\, \bu^{\SB}_{\rm C}(\X(s)) &= \big[({\bf I}- 3\be_{\rm s}\be_{\rm s}^{\rm T})-2({\bf I}+\be_{\rm s}\be_{\rm s}^{\rm T}) \log(\pi\epsilon/4) \big]{\bm f}(s) \\
&\qquad + \int_{\T} \left[ \left(\frac{{\bf I}}{|\overline{\X}|}+ \frac{\overline{\X}\overline{\X}^{\rm T}}{|\overline{\X}|^3}\right){\bm f}(s') - \frac{{\bf I}+\be_{\rm s}(s)\be_{\rm s}(s)^{\rm T} }{|\sin (\pi(s-s'))/\pi|} {\bm f}(s)\right] \, ds',
\end{aligned}
\end{equation}
where $\overline{\X}(s,s') = \X(s)-\X(s')$. We refer to \eqref{KR_SBT} as the Keller--Rubinow SBT, as Keller and Rubinow \cite{keller1976slender} were the first to arrive at (the non-periodic version of) the expression \eqref{KR_SBT}, although they did so by slightly different means. \\

However, one well-documented issue with Keller--Rubinow SBT is that the form of the integral operator in \eqref{KR_SBT} gives rise to high wavenumber instabilities \cite{gotz2000interactions,shelley2000stokesian,tornberg2004simulating,spectral_calc}. Specifically, the eigenvalues of \eqref{KR_SBT} are known to cross zero and switch signs at high frequency, which is nonphysical behavior and leads to noninvertibility of \eqref{KR_SBT}. \\

Various methods have been proposed to eliminate this high wavenumber instability in an asymptotically consistent way. One method, due to Shelley and Ueda \cite{shelley2000stokesian} and Tornberg and Shelley \cite{tornberg2004simulating}, involves inserting a parameter proportional to the fiber radius into the integrand of \eqref{KR_SBT} to regularize the integral kernel, leading to the fiber velocity expression
\begin{equation}\label{TS_SBT}
\begin{aligned}
8\pi\, \bu^{\SB}_{\rm C}(\X(s)) &= \big[({\bf I}- 3\be_{\rm s}\be_{\rm s}^{\rm T})-2({\bf I}+\be_{\rm s}\be_{\rm s}^{\rm T}) \log(\pi\epsilon/4) \big]{\bm f}(s) \\
&\quad + \int_{\T} \bigg[ \frac{{\bf I}+\wh{\overline{\X}}\wh{\overline{\X}}^{\rm T}}{\big(|\overline{\X}|^2+\eta^2\epsilon^2\big)^{1/2}}{\bm f}(s') - \frac{{\bf I}+\be_{\rm s}(s)\be_{\rm s}(s)^{\rm T} }{\big(|\sin (\pi(s-s'))/\pi|^2+\eta^2\epsilon^2\big)^{1/2}} {\bm f}(s)\bigg] \, ds'.
\end{aligned}
\end{equation}
Here $\wh{\overline{\X}}=\overline{\X}/|\overline{\X}|$ and $\eta>\sqrt{e}$ is a parameter chosen such that the operator in \eqref{TS_SBT} is (at least in simple geometries) positive definite. Note that via the integral identities
\begin{align*}
 \int_{-1/2}^{1/2}\bigg(\frac{1}{\abs{\sin(\pi s)/\pi}}-\frac{1}{\abs{s}}\bigg)\, ds = 2\log(4/\pi); \quad \bigg|\int_{-1/2}^{1/2}\frac{1}{\sqrt{s^2+\eta^2\epsilon^2}} ds +2\log(\eta\epsilon)\bigg|\le C\epsilon^2,
 \end{align*}
the second term in the integrand of \eqref{TS_SBT} may be integrated asymptotically to instead yield a local term with coefficient $+2\log(\eta)$ in place of $-2\log(\pi\epsilon/4)$.  \\

More recently, Andersson et al. \cite{norway} propose making use of the fiber integrity condition \eqref{fiber_integrity} along the fiber surface to simply eliminate all $\theta$-dependent terms and obtain the following expression for the fiber velocity:  
\begin{equation}\label{my_SBT}
\begin{aligned}
8\pi\bu^\SB_{\rm C}(\X(s)) &= 2\log(\eta)\, \bm{f}(s) + \int_\T \bigg[\frac{{\bf I}}{(|\overline{\X}|^2+\eta^2\epsilon^2)^{1/2}} + \frac{\overline{\X}\overline{\X}^{\rm T} }{(|\overline{\X}|^2+\epsilon^2)^{3/2}}  \\
&\qquad + \frac{\epsilon^2}{2}\bigg(\frac{{\bf I}}{(|\overline{\X}|^2+\epsilon^2)^{3/2} } - \frac{3\overline{\X}\overline{\X}^{\rm T} }{(|\overline{\X}|^2+\epsilon^2)^{5/2}}\bigg) \bigg] \bm{f}(s') \, ds'.
\end{aligned}
\end{equation}
Here the parameter $\eta\ge1$ can be used to convert \eqref{my_SBT} to a second-kind integral equation, if desired. \\

For $\bx=\X(s)+\epsilon\be_\rho(s,\theta)\in \Gamma_\epsilon$, each of \eqref{KR_SBT}, \eqref{TS_SBT}, and \eqref{my_SBT} satisfy \cite{closed_loop,free_ends,norway} 
\begin{equation}\label{center_est}
\abs{\bu^\SB_{\rm C}(\X(s))-\bu^\SB(\bx)} \le C\epsilon\abs{\log\epsilon} \norm{\bm{f}}_{C^1(\T)}. 
\end{equation}

Moreover, classical SBT is known to closely approximate the velocity field about a 3D filament given by the solution to a well-posed \emph{slender body PDE}, defined by Mori--Ohm--Spirn \cite{closed_loop,free_ends} as the following boundary value problem for the Stokes equations:
\begin{equation}\label{SB_PDE}
\begin{aligned}
-\Delta \bu +\nabla p &= 0, \; \div \bu = 0 \quad \text{in } \Omega_{\epsilon} = \R^3 \backslash \Sigma_{\epsilon}, \\
\int_0^{2\pi} (\bm{\sigma} {\bm n}) \, \mc{J}_{\epsilon}(s,\theta)\, d\theta &= {\bm f}(s) \hspace{1.5cm} \text{ on } \Gamma_{\epsilon}, \\
\bu\big|_{\Gamma_{\epsilon}} &= \bu(s) \hspace{1.5cm} \text{(unknown but independent of }\theta), \\
\abs{\bu} \to 0 & \text{ as } \abs{\bx}\to \infty.
\end{aligned}
\end{equation}
Here $\bm{\sigma}= (\nabla\bu+(\nabla\bu)^{\rm T}) - p{\bf I}$ is the fluid stress tensor, $\bm{n}(\bx)$ is the unit normal vector directed into $\Sigma_\epsilon$ at $\bx\in \Gamma_\epsilon$, and $\mc{J}_\epsilon$ is the Jacobian factor on $\Gamma_\epsilon$ defined in \eqref{Jeps}. In the slender body PDE formulation \eqref{SB_PDE}, the fiber integrity condition is strictly enforced as a partial Dirichlet boundary condition on the fiber surface. The prescribed 1D force density $\bm{f}(s)$ is defined to be the total surface stress over each cross section of the slender body. \\

One reason that the slender body PDE is a physically meaningful framework in which to place SBT is the following energy identity satisfied by solutions to \eqref{SB_PDE}:
\begin{equation}\label{energy_balance}
\int_{\Omega_\epsilon}2\abs{\E(\bu)}^2\, d\bx = \int_\T \bu(s) \bm{f}(s)\, ds, \qquad \E(\bu) = \frac{\nabla \bu + (\nabla \bu)^{\rm T} }{2}.
\end{equation}
Here the viscous dissipation (left hand side) balances the power (right hand side) exerted by the slender filament on the fluid. This identity arises naturally in describing a swimming body in Stokes flow \cite{lauga2020fluid}. \\

In \cite{closed_loop,free_ends}, it is shown that the difference between the velocity field $\bu^\SB(\bx)$ given by classical SBT \eqref{classicalSBT} and the solution $\bu(\bx)$ to \eqref{SB_PDE} satisfies 
\begin{equation}\label{SB_PDE_err1}
\norm{\bu^\SB - \bu}_{L^2(\Omega_\epsilon)} \le C\epsilon\abs{\log\epsilon}\norm{\bm{f}}_{C^1(\T)}.
\end{equation}

Furthermore, the fiber velocity expressions given by any of \eqref{KR_SBT}, \eqref{TS_SBT}, \eqref{my_SBT} satisfy 
\begin{equation}\label{SB_PDE_err2}
\norm{\bu^\SB_{\rm C} - \bu\big|_{\Gamma_\epsilon}}_{L^2(\T)} \le C\epsilon\abs{\log\epsilon}^{3/2}\norm{\bm{f}}_{C^1(\T)}.
\end{equation}

\subsection{Method of regularized Stokeslets}
A different approach to addressing the difficulties associated with the singular kernel arising in classical SBT and, more broadly, in boundary integral methods for Stokes flow, is the method of regularized Stokeslets, developed by Cortez in \cite{cortez2001method} with additional developments in \cite{cortez2005method,ainley2008method,cortez2012slender}. \\

Instead of singular Stokeslets \eqref{stokeslet}, the method relies on regularized Stokeslets, which are derived by solving the Stokes equations in $\R^3$ with a smooth right hand side forcing:
\begin{equation}\label{stokes_reg}
\begin{aligned}
-\Delta\bv+\nabla q &= \bm{f}\phi_\delta(\bx) \\
\div\bv&=0 .
\end{aligned}
\end{equation}
Here $\phi_\delta$ is an approximation to the identity, or \emph{blob function}, whose width is controlled by the parameter $\delta$, i.e.
\begin{equation}\label{phidelta}
\phi_\delta(\bx) = \frac{1}{\delta^3} \phi\bigg(\frac{\bx}{\delta} \bigg); \quad \int_{\R^3}\phi_\delta(\bx) \, d\bx = 1, \quad \lim_{\delta\to 0}\phi_\delta(\bx) = \bm{\delta}(\bx).
\end{equation}
Here we use boldface $\bm{\delta}$ to denote the Dirac delta rather than the parameter $\delta$. The blob function $\phi_\delta$ can be compactly supported or not, and is not required to be strictly positive \cite{nguyen2014reduction,zhao2019method}. For any such $\phi_\delta$, the solution to \eqref{stokes_reg} may be written as $\bv=\frac{1}{8\pi}\mc{S}_\delta(\bx)\bm{f}$, where $\mc{S}_\delta(\bx)$ is known as a regularized Stokeslet and is smooth everywhere in $\R^3$. \\

One of the most common choices of $\phi_\delta$ is
\begin{equation}\label{blob1}
\phi_\delta(\bx) = \frac{15\delta^4}{8\pi(|\bx|^2+\delta^2)^{7/2}},
\end{equation}
because this gives rise to a simple, easy-to-use regularized Stokeslet that looks very similar to the singular Stokeslet \eqref{stokeslet}:
\begin{equation}\label{reg_stokeslet}
\mc{S}_\delta(\bx) = \frac{{\bf I}}{(\abs{\bx}^2+\delta^2)^{1/2}} + \frac{\bx\bx^{\rm T}+\delta^2{\bf I}}{(\abs{\bx}^2+\delta^2)^{3/2}} .
\end{equation}

Many other choices of blob function are possible, including different power laws, exponentially decaying blobs, and compactly supported mollifiers \cite{cortez2015general,zhao2019method}.
The exponentially decaying and compactly supported blobs converge quickly to the singualr Stokeslet in the far field, away from $\bx=0$; however, we will not consider these here as they do not give rise to such a practically useful expression at $\bx=0$ (i.e. along the centerline of the fiber). 
Power law blob functions satisfying additional moment conditions may be constructed which lead to improved convergence to singular Stokeslets in the far field \cite{nguyen2014reduction,zhao2019method}. However, a similar result to Theorem \ref{thm:powerlaw} holds even for these blobs, as we will discuss further in Remark \ref{rem:moment_cond}. As such, for the remainder of this paper we will be considering in detail the regularized Stokeslet of the form \eqref{reg_stokeslet}.\\

Given that regularized Stokeslets are everywhere smooth by construction, it is natural to consider using them in slender body theory as a simple fix for the aforementioned difficulties arising due to the singular nature of \eqref{classicalSBT} along the fiber centerline. Various authors \cite{bouzarth2011modeling1, bouzarth2011modeling2, cortez2012slender, nguyen2014computing, walker2020regularised} have developed and studied versions of SBT based on regularized Stokeslets, although our approach differs in that we consider a regularized doublet correction as in \cite{cortez2012slender, walker2020regularised}, but also place our analysis within the context of the slender body PDE framework and consider the convergence of the flow about the slender body as $\epsilon\to 0$ as well.  \\

As mentioned, to best approximate the flow about $\Sigma_\epsilon$ using regularized Stokeslets, we will still need some form of higher order correction, since \eqref{reg_stokeslet} alone still gives rise to $O(1)$ angular dependence along the surface of the fiber. Thus we also need to make use of regularized doublets. In analogy to \eqref{doublet}, we will define a regularized doublet as the 1/2 Laplacian of a regularized Stokeslet. Note that this definition differs slightly from \cite{cortez2012slender,cortez2015general,nguyen2014computing}, where the doublet is defined as the negative Laplacian of a regularized Stokeslet. Here we use $+1/2$ rather than $-1$ for ease of comparison with classical SBT; note that due to the choice of doublet coefficient \eqref{doub_coeff} the final expression \eqref{MRS_SBT} for regularized SBT is exactly the same as in \cite{cortez2012slender,nguyen2014computing}. \\

It turns out that to best serve the purpose of asymptotically cancelling certain unwanted $O(1)$ terms arising from the regularized Stokeslet in SBT, it often makes sense to use a different blob function to derive the regularized doublet \cite{cortez2012slender,ainley2008method}. In particular, the regularized doublet corresponding to \eqref{reg_stokeslet} for the purposes of SBT is given by
\begin{equation}\label{reg_doublet}
\mc{D}_\delta(\bx) = \frac{{\bf I}}{(\abs{\bx}^2+\delta^2)^{3/2}} - 3\frac{\bx\bx^{\rm T}+\delta^2{\bf I}}{(\abs{\bx}^2+\delta^2)^{5/2}} ,
\end{equation}
which arises from the blob function 
\begin{align*}
\phi_\delta = \frac{3\delta^2}{4\pi (|\bx|^2+\delta^2)^{5/2}}.
\end{align*}

Again, the purpose of a doublet correction in SBT is to asymptotically cancel undesired $O(1)$ terms at the fiber surface which arise from integrating the regularized Stokeslet over the fiber centerline. In classical SBT, there is just one such undesired term, which arises due to the $\theta$-dependence of $\bx\bx^{\rm T}$ on $\Gamma_\epsilon$. In regularized SBT, different choices of regularized Stokeslets give rise to additional $O(1)$ terms which require cancellation. In the case of \eqref{reg_stokeslet}, for example, the additional $\delta^2{\bf I}$ term also requires cancellation. \\

In analogy to classical SBT \eqref{classicalSBT}, we may then consider the following approximation to the flow about the slender body $\Sigma_\epsilon$ using the regularized Stokeslet \eqref{reg_stokeslet} and regularized doublet \eqref{reg_doublet} \cite{cortez2012slender}:  
\begin{align}\label{MRS_SBT}
\bu^\delta(\bx) = \frac{1}{8\pi}\int_\T \bigg(\mc{S}_\delta(\bx-\X(s')) + \beta(\bx)&\mc{D}_\delta(\bx-\X(s')) \bigg) \bm{f}(s') \, ds',  \\
\label{doub_coeff}
\beta(\bx) &= \begin{cases}
\frac{\epsilon^2+\delta^2}{2}, & \bx\in \overline{\Omega}_\epsilon \\
\frac{\delta^2}{2}, & \bx=\X(s)\in \Gamma_0.
\end{cases}
\end{align}
Now the kernel of \eqref{MRS_SBT} is smooth along the fiber centerline; in particular, to obtain an expression for the velocity of the filament itself, we can simply evaluate \eqref{MRS_SBT} at $\bx=\X(s)$. Note that in order to achieve the desired cancellation properties with the regularized doublet, the doublet coefficient $\beta(\bx)$ must be different depending on whether we are considering points $\bx$ along the fiber centerline or in the bulk fluid, up to and including the surface of the actual 3D fiber. This will be discussed further in Remark \ref{rem:doublet}. \\

In addition to the velocity approximation \eqref{MRS_SBT}, a regularized approximation to the fluid pressure field about the slender body may also be defined. For the choice of blob function \eqref{blob1}, the regularized slender body pressure approximation is given by \cite{cortez2005method}
\begin{equation}\label{pressure_reg}
p^\delta(\bx) = \frac{1}{4\pi}\int_{\T}\bigg(\frac{\bm{R}\cdot\bm{f}(s')}{(|\bm{R}|^2+\delta^2)^{3/2}} + \frac{3\delta^2\bm{R}\cdot\bm{f}(s')}{2(|\bm{R}|^2+\delta^2)^{5/2}} \bigg)\,ds' . 
\end{equation}
Note that we have rewritten the pressure tensor from \cite{cortez2005method} to isolate the $\delta^2$ term for ease of comparison with the non-regularized counterpart. \\

 For a complete comparison between SBT based on regularized Stokeslets \eqref{MRS_SBT} and its classical counterpart \eqref{classicalSBT}, we will also consider the force-per-unit-length along the surface of the fiber due to the regularized Stokeslet velocity and pressure fields. For $s\in \T$, we define the force density $\bm{f}^\delta(s)$ analogously to \eqref{SB_PDE} as
\begin{equation}\label{reg_force}
\bm{f}^\delta(s) = \int_0^{2\pi}\bigg( 2\E(\bu^\delta)\bm{n} - p^\delta\bm{n} \bigg)\mc{J}_\epsilon(s,\theta)d\theta.
\end{equation}
Here $\bm{n}(s,\theta)$ is the unit vector normal to the fiber surface at $\bx(s,\theta)\in\Gamma_\epsilon$ which points into the slender body. Each of $\E(\cdot)$, $\bu^\delta$, $p^\delta$, and $\mc{J}_\epsilon$ are as defined in \eqref{energy_balance}, \eqref{MRS_SBT}, \eqref{pressure_reg}, and \eqref{Jeps}, respectively. \\

Note that the analogous force-per-unit-length due to classical SBT, which we denote $\bm{f}^\SB$, may be defined similarly to \eqref{reg_force} without issue, since the expressions are evaluated on the fiber surface rather than centerline and thus are not singular. From \cite{closed_loop}, we have that $\bm{f}^\SB$ satisfies 
\begin{equation}\label{fSB_est}
 \abs{\bm{f}(s) -\bm{f}^\SB(s)} \le C\epsilon\norm{\bm{f}}_{C^1(\T)}.
 \end{equation} 
The estimate \eqref{fSB_est} is an essential ingredient in proving the convergence estimates \eqref{SB_PDE_err1}, \eqref{SB_PDE_err2} to the slender body PDE solution, and therefore we are also interested in determining what type of estimate holds in the case of regularized Stokeslets. \\


\subsection{Statement of main result}
Our main result is the following.
\begin{theorem}\label{thm:powerlaw}
Consider a slender body $\Sigma_\epsilon$ satisfying the geometric constraints of Section \ref{sec:geomANDsbt}. Given a line force density $\bm{f}\in C^1(\T)$, let $\bu^\delta(\bx)$ be the slender body approximation \eqref{MRS_SBT} due to the method of regularized Stokeslets, using \eqref{reg_stokeslet} and \eqref{reg_doublet} as the regularized Stokeslet and doublet, respectively. Let $\bu^\SB(\bx)$ and $\bu^\SB_{\rm C}(\X(s))$ be the bulk and centerline velocity expressions, respectively, due to classical SBT (see \eqref{classicalSBT} and \eqref{KR_SBT}).\\

Along the fiber centerline $\X(s) = \bx\in\Gamma_0$, the difference between $\bu^\delta$ and $\bu^\SB_{\rm C}$ satisfies
\begin{equation}\label{centerline_est}
\abs{8\pi \big(\bu^\delta(\X(s)) - \bu^\SB_{\rm C}(\X(s)) \big) + 2\log\bigg(\frac{\delta}{\epsilon}\bigg)({\bf I}+\be_{\rm s}\be_{\rm s}^{\rm T})\bm{f}(s) } \le C(\epsilon\abs{\log\epsilon}+\delta\abs{\log\delta})\norm{\bm{f}}_{C^1(\T)}.
\end{equation}
In particular, if the regularization $\delta$ is proportional to $\epsilon$, then in fact $\delta\equiv\epsilon$ is necessary to eliminate an $O(1)$ discrepancy between the fiber velocity approximations.\\

However, for $\bx\in \Gamma_\epsilon$, along the actual surface of the fiber, the difference $\bu^\delta - \bu^\SB$ satisfies
\begin{equation}\label{bulk_est}
 \abs{8\pi \big(\bu^\delta(\bx) - \bu^\SB(\bx) \big) + \log\bigg(1+\frac{\delta^2}{\epsilon^2}\bigg)({\bf I}+\be_{\rm s}\be_{\rm s}^{\rm T}) \bm{f}(s)} \le C\sqrt{\epsilon^2+\delta^2}\log(\epsilon^2+\delta^2)\norm{\bm{f}}_{C^1(\T)}.
\end{equation}
Here there is an $O(1)$ discrepancy that cannot be eliminated by any choice of $\delta \propto \epsilon$. \\

Furthermore, along the fiber surface $\Gamma_\epsilon$, the regularized force density $\bm{f}^\delta$ \eqref{reg_force} satisfies
\begin{equation}\label{force_est}
\abs{\bm{f}^\delta(s) - \frac{\bm{f}(s)}{(1+\delta^2/\epsilon^2)^2} } \le C \sqrt{\epsilon^2+\delta^2}\norm{\bm{f}}_{C^1(\T)} .
\end{equation}

Finally, for any fluid point $\bx$ which is a distance $\epsilon^{1-\alpha}$, $0<\alpha<1$, away from the fiber centerline, if $\bx$ can be uniquely represented in the fiber frame \eqref{frameODE} (i.e. there exists unique closest point to $\bx$ on the fiber centerline), then
\begin{equation}\label{decay_est}
\abs{\bu^\delta(\bx)- \bu^\SB(\bx)} \le \frac{C\delta^2}{\epsilon^{2(1-\alpha)}}\norm{\bm{f}}_{C^0(\T)}.
\end{equation}
In particular, for $\delta\propto\epsilon$, we obtain $O(\epsilon^{2\alpha})$ convergence.

\end{theorem}

The main takeaways of Theorem \ref{thm:powerlaw} can be summarized estimate-by-estimate as follows. \\

1. (\emph{Estimate \eqref{centerline_est}}) The regularized Stokeslet approximation $\bu^\delta(\X(s))$ for the velocity of the fiber itself can be made to agree asymptotically with the classical Keller--Rubinow expression \eqref{KR_SBT} (and hence the fiber velocity given by the slender body PDE as well, due to estimate \eqref{SB_PDE_err2}), but only for the choice of regularization parameter $\delta=\epsilon$. In particular, choosing $\delta=c\epsilon$, $c\neq 1$, leads to an $O(1)$ discrepancy between classical SBT and the method of regularized Stokeslets as $\epsilon\to 0$ -- specifically, a difference of the form $\frac{1}{4\pi}\log(c)({\bf I}+\be_{\rm s}\be_{\rm s}^{\rm T})\bm{f}(s)$. The necessity of choosing $\delta=\epsilon$ here stems from the fact that the leading order behavior of SBT scales like $\log\epsilon$ and is thus very sensitive to small perturbations. \\

This behavior may be contrasted with that of the parameter $\eta$ in the Tornberg--Shelley formulation \eqref{TS_SBT} and also in Andersson \cite{norway}. Note that in both formulations, $\eta$ may be taken to be larger than 1 (in fact, $\eta>\sqrt{e}$ is used in \eqref{TS_SBT}) and both velocity expressions still converge to the slender body PDE solution \eqref{SB_PDE} by \eqref{SB_PDE_err2}. This is because the effect of $\eta>1$ is cancelled by the local $\log(\eta)$ term in \eqref{my_SBT} or by the $\eta$ in the denominator of the second term of the integrand in \eqref{TS_SBT} (note that this term can be integrated asymptotically to yield a similar local $\log(\eta)$ term). In particular, the leading order effect of inserting the parameter $\eta$ is cancelled, which preserves tha asymptotic consistency with Keller--Rubinow SBT \eqref{KR_SBT} for values of $\eta$ besides $\eta=1$. \\

2. (\emph{Estimate \eqref{bulk_est}}) On the other hand, at the surface $\Gamma_\epsilon$ of the slender body, there is always an $O(1)$ discrepancy between classical SBT and the method of regularized Stokeslets for any choice of $\delta$ proportional to $\epsilon$ of the form $\frac{1}{8\pi}\log\big(1+\frac{\delta^2}{\epsilon^2}\big)({\bf I}+\be_{\rm s}\be_{\rm s}^{\rm T}) \bm{f}(s)$. In particular, due to takeaway $\#1$,  the fiber velocity $\bu^\delta(\X(s))$ actually differs from the velocity at the fiber surface by an $O(1)$ quantity, i.e. there is no sense of 3D structure to the filament, and the slender body PDE assumption \eqref{SB_PDE} of a constant velocity across fiber cross sections is not even approximately satisfied. \\

Again, since the leading order term of SBT blows up like $\log\epsilon$ for $\epsilon$ small, other methods which capture this leading order behavior may still differ by an $O(1)$ quantity. The presence of an additional $\delta^2$ in each denominator of \eqref{reg_stokeslet}, \eqref{reg_doublet} when evaluating on the fiber surface $\Gamma_\epsilon$ (i.e. the denominators look like $|\overline\X|^2+\epsilon^2+\delta^2$) means that the surrounding fluid behaves as if the fiber actually has a thickness of $\sqrt{\epsilon^2+\delta^2}$ rather than $\epsilon$. Because of the logarithmic leading order behavior (coming from integrating $1/(|\overline\X|^2+\epsilon^2+\delta^2)^{1/2}$ along the centerline), this is not a small perturbation even when $\delta\propto\epsilon$. Note that this discrepancy is also apparent in the asymptotic analysis in Cortez--Nicholas \cite{cortez2012slender}, where the authors perform an asymptotic expansion of \eqref{MRS_SBT}, evaluated at the fiber surface $\Gamma_\epsilon$, about $\epsilon=0$ and obtain the same logarithmic disagreement with classical SBT \eqref{KR_SBT} in the local term. Note that the claim following \cite{cortez2012slender}, equation (23), that this disagreement can be cancelled for a periodic fiber with a particular choice of $\delta\sim \epsilon$ is not true, as the resulting expression does not agree with the \emph{periodic} Keller--Rubinow expression \eqref{KR_SBT}, where the local logarithmic coefficient should be of the form $-2\log(\pi\epsilon/4)$. This discrepancy can, however, be avoided by using the centerline evaluation of \eqref{MRS_SBT} as the fiber velocity instead. \\

3. (\emph{Estimate \eqref{force_est}}) The nonconvergence of regularized SBT to both classical SBT and the slender body PDE as $\epsilon\to0$ is further manifested in the force estimate \eqref{force_est} along $\Gamma_\epsilon$. Instead of recovering the prescribed force density $\bm{f}(s)$ at the fiber surface as in \eqref{fSB_est}, the method of regularized Stokeslets spreads the prescribed force density $\bm{f}$ into the surrounding fluid domain, giving rise to an additional body force in the fluid \cite{zhao2019method}. If $\delta=\epsilon$, for example, then $1/4$ of $\bm{f}$ is recovered at the surface of the filament, while the rest of the forcing is spread into the fluid bulk. In particular, SBT based on the method of regularized Stokeslets does not (asymptotically) satisfy the energy identity \eqref{energy_balance} associated with swimming in Stokes flow. \\

4. (\emph{Estimate \eqref{decay_est}}) Despite the $O(1)$ discrepancy at the fiber surface, due to \eqref{decay_est} we have that the difference in velocity expressions at any fluid point $\bx$ near but not on the surface of the fiber actually does converge to 0 as $\epsilon\to 0$. The rate of convergence (in $\epsilon$) goes to zero as $\bx$ approaches the fiber surface. Due to this decay estimate \eqref{decay_est} and to the small magnitude of the $O(1)$ difference $\bu^\delta-\bu^\SB$ at the fiber surface (note that when $\delta=\epsilon$, for example, this difference is $\frac{1}{4\pi}\log(2)\approx 0.05516$), the discrepancy between the flows due to regularized and classical SBT may not have a qualitatively noticeable effect in practice.
In particular, numerical demonstrations in Section \ref{sec:numerics} indicate that even when multiple fibers interact hydrodynamically, the discrepancy between regularized and classical SBT may not be large enough to produce significantly different velocity fields around the fibers. \\

\begin{remark}[Doublet coefficient]\label{rem:doublet}
We would like to address the different choice of doublet coefficient for the bulk versus centerline expressions in regularized SBT \eqref{MRS_SBT}. If the same coefficient is used for both expressions, then there will be an additional $O(1)$ quantity appearing in either the centerline or bulk expression (depending on which coefficient is used) that is not cancelled asymptotically. For example, if the centerline doublet coefficient $\frac{\delta^2}{2}$ is used in the bulk, then a term
\[ -\frac{\epsilon^2}{2}\int_\T\mc{D}_\delta(\bx-\X(s')) \, \bm{f}(s')\, ds' \]
must be added to the $\bu^\delta(\bx)$ given in \eqref{MRS_SBT}.
Using Lemma \ref{CPAM_lemmas}, which will be introduced in Section \ref{sec:proof}, this term can be shown to contribute an $O(1)$ term of the form 
\[\frac{1}{(\epsilon^2+\delta^2)^2}\big( -(\epsilon^4-\epsilon^2\delta^2){\bf I}+(\epsilon^4+\epsilon^2\delta^2)\be_{\rm s}\be_{\rm s}^{\rm T}+2\epsilon^4\be_\rho\be_\rho^{\rm T}\big)\bm{f}(s)\]
 at the fiber surface. Although the ${\bf I}$ term can be eliminated by choosing $\delta=\epsilon$, this term still gives rise to an even greater discrepancy than \eqref{bulk_est} with classical SBT. Similarly, if the bulk coefficient $\frac{\delta^2+\epsilon^2}{2}$ is used instead along the fiber centerline, then an additional term
\[ \frac{\epsilon^2}{2}\int_\T\mc{D}_\delta(\overline{\X}(s,s')) \, \bm{f}(s')\, ds' \]
is added to the centerline expression \eqref{MRS_SBT}, which can be shown (again via Lemma \ref{CPAM_lemmas} in Section \ref{sec:proof}) to have an $O(1)$ contribution 
\[\frac{1}{(\epsilon^2+\delta^2)^2}\big( (\epsilon^4-\epsilon^2\delta^2){\bf I}-(\epsilon^4+\epsilon^2\delta^2)\be_{\rm s}\be_{\rm s}^{\rm T}\big)\bm{f}(s) \] 
to the fiber velocity. Again, only the ${\bf I}$ term may be eliminated by choosing $\delta=\epsilon$. Note that the need to choose different doublet coefficients removes one of the apparent benefits of regularized SBT, however; namely, the same kernel cannot be used to describe both the fiber velocity and the flow around it.   
 \end{remark}

\begin{remark}[Other blob functions]\label{rem:moment_cond}
The results of Theorem \ref{thm:powerlaw} hold even for power law blobs satisfying additional moment conditions which lead to better convergence (in terms of $\delta$) to singular Stokeslets in the far field. For example, we may consider the blob 
\begin{equation}\label{blob2}
\phi_\delta(\bx) = \frac{15\delta^4(5\delta^2-2\abs{\bx}^2)}{16\pi(\abs{\bx}^2+\delta^2)^{9/2}},
\end{equation}
which gives rise to the regularized Stokeslet 
\begin{equation}\label{reg_stokeslet2}
\mc{S}_\delta(\bx) = \frac{{\bf I}}{(\abs{\bx}^2+\delta^2)^{1/2}} + \frac{\bx\bx^{\rm T}}{(\abs{\bx}^2+\delta^2)^{3/2}} + \frac{4\delta^4{\bf I}+\delta^2\abs{\bx}^2{\bf I}+3\delta^2\bx\bx^{\rm T}}{2(\abs{\bx}^2+\delta^2)^{5/2}}.
\end{equation}

The regularized doublet corresponding to \eqref{reg_stokeslet2} (in the sense of the cancellation properties discussed following equation \eqref{reg_doublet}) is given by 
\begin{equation}\label{reg_doublet2}
\mc{D}_\delta(\bx) = \frac{{\bf I}}{(\abs{\bx}^2+\delta^2)^{3/2}} - \frac{3\bx\bx^{\rm T}}{(\abs{\bx}^2+\delta^2)^{3/2}} - \frac{12\delta^4{\bf I}-3\delta^2\abs{\bx}^2{\bf I}+15\delta^2\bx\bx^{\rm T}}{2(\abs{\bx}^2+\delta^2)^{7/2}},
\end{equation}
and arises from the blob function
\begin{align*}
\phi_\delta(\bx) = \frac{15\delta^4}{8\pi(|\bx|^2+\delta^2)^{7/2}}.
\end{align*}
Note that we are again defining the regularized doublet as the $+1/2$ Laplacian of the regularized Stokeslet, rather than the $-1$ Laplacian.
Both of the above blobs are constructed to have a vanishing second moment, i.e. 
\[\int_0^\infty r^4\phi_\delta(r)\, dr=0, \]
 which has been shown to yield regularized Stokeslets and doublets with improved convergence to their non-regularized counterparts as $\bx\to\infty$ \cite{nguyen2014reduction,zhao2019method}. 
We may consider using \eqref{reg_stokeslet2} and \eqref{reg_doublet2} in the regularized SBT expression \eqref{MRS_SBT} in place of \eqref{reg_stokeslet} and \eqref{reg_doublet}.
 However, the logarithmic discrepancy noted in Theorem \ref{thm:powerlaw} remains: it is fundamentally a local issue rather than a decay issue, stemming from the additional $\delta$ present in the regularized Stokeslet denominator when evaluating at the fiber surface. It is the behavior of regularized Stokeslets at distances $O(\epsilon)$ from the fiber centerline that causes disagreement between regularized and classical SBT.
Using a blob function that is compactly supported within the slender body would likely remedy this issue, but at the expense of simplicity and some practical utility. \\
\end{remark}

The remainder of the paper is structured as follows. In Section \ref{sec:proof}, we prove Theorem \ref{thm:powerlaw}, beginning in Section \ref{subsec:prelims} with some preliminary definitions and a key lemma based on \cite{closed_loop}. We then proceed to show each of the velocity estimates in Section \ref{subsec:velocity} and the force estimate in Section \ref{subsec:force}. We conclude with a numerical study of the difference between regularized and classical SBT to assess the practical implications of Theorem \ref{thm:powerlaw}.


\section{Proof of Theorem \ref{thm:powerlaw}}\label{sec:proof}
\subsection{Preliminaries}\label{subsec:prelims}
We begin by defining some notation and a crucial lemma for estimating the integral expressions appearing in \eqref{MRS_SBT} and \eqref{classicalSBT}. \\

For $s$, $s'\in \T$, it will be useful to work in terms of the difference $\bars = s-s'$ rather than $s'$. Then, for $\bx(s,\theta) = \X(s) +\epsilon \be_\rho(s,\theta)$ along $\Gamma_\epsilon$, let 
\begin{equation}\label{Rdef0}
\bm{R}(s,\bars,\theta)=\bx -\X(s-\bars)=\overline{\X}(s,s-\bars)+\epsilon\be_\rho(s,\theta), 
\end{equation}
where $\overline{\X}(s,s-\bars)=\X(s)-\X(s-\bars)$. Since we assume $\X\in C^2$, we may make use of the following expansion for $\bm{R}$: 
\begin{equation}\label{Rdef}
 \bm{R}(s,\bars,\theta)= \bars \be_{\rm s}+\epsilon \be_\rho(s,\theta) + \bars^2\bm{Q} \quad \text{ for some }\abs{\bm{Q}}\le \frac{\kappa_{\max}}{2}.
 \end{equation} 

Throughout the following, for any function $\bm{g}\in C^1(\T)$, we will denote
\begin{align*}
\norm{\bm{g}}_{C^1(\T)} = \norm{\bm{g}}_{C^0(\T)} + \norm{\bm{g}'}_{C^0(\T)}, \text{ where } \norm{\bm{g}}_{C^0(\T)} = \max_{s\in\T}\abs{\bm{g}(s)}.
\end{align*}

The proof of Theorem \ref{thm:powerlaw} relies essentially entirely on the following lemma, which is a simple adaptation of Lemmas 3.3--3.6 in \cite{closed_loop}. 
\begin{lemma}\label{CPAM_lemmas}
Let $\bm{R}$ be as defined in \eqref{Rdef}. 
For $m,n$ integers with $m\ge0$, $n>0$, and for $\epsilon,\delta$ sufficiently small, we have 
\begin{equation}\label{lemma3_3}
\int_{-1/2}^{1/2}\frac{\abs{\bars}^m}{(|\bm{R}|^2+\delta^2)^{n/2}} \le 
\begin{cases}
C\abs{\log(\epsilon^2+\delta^2)}, & n=m+1 \\
C(\epsilon^2+\delta^2)^{(m+1-n)/2} , & n\ge m+2.
\end{cases}
\end{equation}

For $\bm{g}\in C^1(\T)$, if $m>0$ is odd and $n\ge m+2$, we may obtain the following refinement: 
\begin{equation}\label{lemma3_4}
\abs{\int_{-1/2}^{1/2}\frac{\bars^m}{(|\bm{R}|^2+\delta^2)^{n/2}}\bm{g}(\bars) d\bars} \le 
\begin{cases}
C\norm{\bm{g}}_{C^1(\T)}\abs{\log(\epsilon^2+\delta^2)}, & n=m+1 \\
C\norm{\bm{g}}_{C^1(\T)}(\epsilon^2+\delta^2)^{(m+2-n)/2}, & n\ge m+3
\end{cases}
\end{equation}

If $m\ge 0$ is even and $n\ge m+3$, we have
\begin{equation}\label{lemma3_5}
\abs{\int_{-1/2}^{1/2}\frac{\bars^m}{(|\bm{R}|^2+\delta^2)^{n/2}}\bm{g}(\bars) d\bars - (\epsilon^2+\delta^2)^{(m+1-n)/2} d_{mn} \bm{g}(0)} \le C(\epsilon^2+\delta^2)^{(m+2-n)/2}\norm{\bm{g}}_{C^1(\T)}, 
\end{equation}
where
\begin{align*}
d_{mn} &= \int_{-\infty}^\infty\frac{\tau^m}{(\tau^2+1)^{n/2}}d\tau; \\
d_{03} &= 2, \; d_{05} =\frac{4}{3}, \; d_{07} = \frac{16}{15}, \; d_{25} = \frac{2}{3}, \; d_{27}=\frac{4}{15}.
\end{align*}

Finally, if $n=1$ or $n=3$, we have
\begin{equation}\label{lemma3_6}
\begin{aligned}
&\abs{\int_{-1/2}^{1/2}\frac{\bars^{n-1}}{(|\bm{R}|^2+\delta^2)^{n/2}}\bm{g}(\bars)d\bars - \int_{-1/2}^{1/2}\bigg(\frac{\bars^{n-1}}{|\overline{\X}|^n}\bm{g}(\bars) -\frac{\bm{g}(0)}{|\bars|} \bigg)d\bars +\bm{g}(0)\log(\epsilon^2+\delta^2) + (n-1)\bm{g}(0)} \\
&\hspace{4cm} \le C\norm{\bm{g}}_{C^1(\T)}(\epsilon^2+\delta^2)^{1/2}\abs{\log(\epsilon^2+\delta^2) }
\end{aligned}
\end{equation}

In each of the above bounds, the constant $C$ depends only on $n$, $c_\Gamma$, and $\kappa_{\max}$.
\end{lemma}

\begin{proof}
The estimates \eqref{lemma3_3}, \eqref{lemma3_4}, \eqref{lemma3_5}, and \eqref{lemma3_6} are straightforward adaptations of Lemmas 3.3, 3.4, 3.5, and 3.6, respectively, in \cite{closed_loop}, to include a factor of $\delta^2$ in the denominator. 
\end{proof}
Note that Lemma \ref{CPAM_lemmas} may also be applied with $\delta=0$, which is the case covered in \cite{closed_loop}. 

\subsection{Velocity differences}\label{subsec:velocity}
To prove Theorem \ref{thm:powerlaw}, we begin by showing the surface estimate \eqref{bulk_est} for the difference $\bu^\delta-\bu^\SB$ along $\Gamma_\epsilon$ in full detail. We then show that the centerline estimate \eqref{centerline_est} and the bulk (decay) estimate \eqref{decay_est} can be proved following the same outline.  \\

At any $\bx=\bx(s,\theta)$ on the fiber surface $\Gamma_\epsilon$, we write the difference $\bu^\delta\big(\bx(s,\theta)\big) - \bu^\SB\big(\bx(s,\theta)\big)$ as
\begin{align*}
\bu^\delta\big(\bx(s,\theta)\big) &- \bu^\SB\big(\bx(s,\theta)\big) = \frac{1}{8\pi}\bigg(S_1 + S_2 + S_3 + D_1 + D_2 + D_3\bigg); \\
S_1 &= \int_{-1/2}^{1/2} \bigg(\frac{1}{(|\bm{R}|^2+\delta^2)^{1/2}}-\frac{1}{|\bm{R}|} \bigg) \bm{f}(s-\bars) \, d\bars \\
S_2 &= \int_{-1/2}^{1/2} \bigg(\frac{1}{(|\bm{R}|^2+\delta^2)^{3/2}} - \frac{1}{|\bm{R}|^3} \bigg)(\bm{R}\bm{R}^{\rm T}) \bm{f}(s-\bars) \, d\bars \\
S_3 &= \int_{-1/2}^{1/2} \frac{\delta^2}{(|\bm{R}|^2+\delta^2)^{3/2}} \bm{f}(s-\bars) \, d\bars \\
D_1 &= \frac{1}{2}\int_{-1/2}^{1/2} \bigg(\frac{\epsilon^2+\delta^2}{(|\bm{R}|^2+\delta^2)^{3/2}}-\frac{\epsilon^2}{|\bm{R}|^3} \bigg) \bm{f}(s-\bars) \, d\bars \\
D_2 &= -\frac{3}{2}\int_{-1/2}^{1/2} \bigg(\frac{\epsilon^2+\delta^2}{(|\bm{R}|^2+\delta^2)^{5/2}} - \frac{\epsilon^2}{|\bm{R}|^5} \bigg) (\bm{R}\bm{R}^{\rm T})\bm{f}(s-\bars) \, d\bars \\
D_3 &= -\frac{3}{2}\int_{-1/2}^{1/2} \frac{\delta^2(\epsilon^2+\delta^2)}{(|\bm{R}|^2+\delta^2)^{5/2}} \bm{f}(s-\bars) \, d\bars 
\end{align*}
We proceed to estimate each term separately.
By Lemma \ref{CPAM_lemmas}, subtracting estimate \eqref{lemma3_6} with $\delta=0$ from \eqref{lemma3_6} with $\delta>0$, the difference $S_1$ satisfies 
\begin{align*}
\abs{S_1 + \log\bigg(1+\frac{\delta^2}{\epsilon^2}\bigg)\bm{f}(s)} &= \abs{S_1 + \big(\log(\epsilon^2+\delta^2)-\log(\epsilon^2)\big)\bm{f}(s)} \\
&\le C\norm{\bm{f}}_{C^1(\T)}(\epsilon^2+\delta^2)^{1/2}\abs{\log(\epsilon^2+\delta^2) }.
\end{align*}

For $S_2$, we use \eqref{Rdef} to write 
\begin{equation}\label{RRT}
\begin{aligned}
\bm{RR}^{\rm T} &= R_1 + R_2 + R_3; \\
R_1 &= \bars^2\be_{\rm s}\be_{\rm s}^{\rm T}+\epsilon^2\be_\rho\be_\rho^{\rm T} \\
R_2 &= \bars\epsilon(\be_{\rm s}\be_\rho^{\rm T} + \be_\rho\be_{\rm s}^{\rm T}) +\epsilon\bars^2(\be_\rho\bm{Q}^{\rm T}+ \bm{Q}\be_\rho^{\rm T}) \\
R_3 &= \bars^3(\be_{\rm s}\bm{Q}^{\rm T}+\bm{Q}\be_{\rm s}^{\rm T}) + \bars^4\bm{Q}\bm{Q}^{\rm T}
\end{aligned}
\end{equation}
and also note 
\begin{equation}\label{Rdiff}
\begin{aligned}
\abs{\frac{1}{(|\bm{R}|^2+\delta^2)^{n/2}} - \frac{1}{|\bm{R}|^n}}
&= \abs{\frac{-\delta^2}{|\bm{R}|\sqrt{|\bm{R}|^2+\delta^2}(|\bm{R}|+\sqrt{|\bm{R}|^2+\delta^2})}\sum_{\ell = 0}^{n-1}\frac{1}{(|\bm{R}|^2+\delta^2)^{\ell/2}|\bm{R}|^{n-1-\ell}} } \\
&\le \frac{\delta^2}{\abs{\bm{R}}^{n+2}}.
\end{aligned}
\end{equation}

Then, using Lemma \ref{CPAM_lemmas}, estimate \eqref{lemma3_3}, along with equation \eqref{Rdiff}, we have
\begin{align*}
\abs{\int_{-1/2}^{1/2} \bigg(\frac{1}{(|\bm{R}|^2+\delta^2)^{3/2}} - \frac{1}{|\bm{R}|^3} \bigg)R_3 \, \bm{f}(s-\bars) \, d\bars} &\le \frac{C\delta^2}{(\epsilon^2+\delta^2)^{1/2}}\norm{\bm{f}}_{C^0(\T)}.
\end{align*}

Furthermore, using Lemma \ref{CPAM_lemmas}, estimates \eqref{lemma3_4} and \eqref{lemma3_3}, along with \eqref{Rdiff}, we obtain
\begin{align*}
\abs{\int_{-1/2}^{1/2} \bigg(\frac{1}{(|\bm{R}|^2+\delta^2)^{3/2}} - \frac{1}{|\bm{R}|^3} \bigg)R_2 \, \bm{f}(s-\bars) \, d\bars} &\le \frac{C\delta^2\epsilon}{\epsilon^2+\delta^2}\norm{\bm{f}}_{C^1(\T)}.
\end{align*}

Finally, subtracting the estimates \eqref{lemma3_6} and \eqref{lemma3_5} with $\delta=0$ from \eqref{lemma3_6} and \eqref{lemma3_5} with $\delta>0$, we obtain 
\begin{align*}
&\abs{\int_{-1/2}^{1/2} \bigg(\frac{1}{(|\bm{R}|^2+\delta^2)^{3/2}} - \frac{1}{|\bm{R}|^3} \bigg)R_1 \, \bm{f}(s-\bars) \, d\bars + \log(1+ \delta^2/\epsilon^2)\be_{\rm s}(\be_{\rm s}\cdot\bm{f}(s)) + \frac{2\delta^2}{\epsilon^2+\delta^2}\be_\rho(\be_\rho\cdot\bm{f}(s))} \\
 &\hspace{3cm}\le C\bigg(\sqrt{\epsilon^2+\delta^2}\abs{\log\big(\epsilon^2+\delta^2\big)}+\frac{\epsilon^2}{(\delta^2+\epsilon^2)^{1/2}} \bigg)\norm{\bm{f}}_{C^1(\T)}.
\end{align*}

Altogether, the term $S_2$ satisfies
\begin{align*}
\abs{S_2 + \log\bigg(1+\frac{\delta^2}{\epsilon^2}\bigg)\be_{\rm s}(\be_{\rm s}\cdot\bm{f}(s)) + \frac{2\delta^2}{\epsilon^2+\delta^2}\be_\rho(\be_\rho\cdot\bm{f}(s)) } &\le C(\epsilon^2+\delta^2)^{1/2}\abs{\log\big(\epsilon^2+\delta^2\big)}\norm{\bm{f}}_{C^1(\T)} .
\end{align*}

Next, using Lemma \ref{CPAM_lemmas}, estimate \eqref{lemma3_5}, the final Stokeslet term $S_3$ satisfies 
\begin{align*}
\abs{S_3 - \frac{2\delta^2}{\epsilon^2+\delta^2}\bm{f}(s)} &\le \frac{C\delta^2}{(\epsilon^2+\delta^2)^{1/2}}\norm{\bm{f}}_{C^1(\T)}. 
\end{align*}

We now turn to the doublet terms. Using Lemma \ref{CPAM_lemmas}, subtracting estimate \eqref{lemma3_5} with $\delta=0$ from \eqref{lemma3_5} with $\delta>0$, the $O(1)$ terms cancel to leave 
\begin{align*}
\abs{D_1} &\le C\bigg((\epsilon^2+\delta^2)^{1/2}+\frac{\epsilon^2}{\sqrt{\epsilon^2+\delta^2}} \bigg)\norm{\bm{f}}_{C^1(\T)} \le C(\epsilon^2+\delta^2)^{1/2}\norm{\bm{f}}_{C^1(\T)}.
\end{align*}

To estimate $D_2$, we again use the expansion \eqref{RRT} of $\bm{RR}^{\rm T}$ as well as the bound \eqref{Rdiff}. By Lemma \ref{CPAM_lemmas}, equation \eqref{lemma3_3}, we have
\begin{align*}
\abs{-\frac{3}{2}\int_{-1/2}^{1/2} \bigg(\frac{\epsilon^2+\delta^2}{(|\bm{R}|^2+\delta^2)^{5/2}} - \frac{\epsilon^2}{|\bm{R}|^5} \bigg) R_3 \, \bm{f}(s-\bars) \, d\bars} &\le \frac{C\delta^2}{(\delta^2+\epsilon^2)^{1/2}}\bigg(\frac{\epsilon^2}{\delta^2+\epsilon^2} +1 \bigg)\norm{\bm{f}}_{C^0(\T)}.
\end{align*}
Furthermore, using equations \eqref{lemma3_4} and \eqref{lemma3_3}, we have 
\begin{align*}
\abs{-\frac{3}{2}\int_{-1/2}^{1/2} \bigg(\frac{\epsilon^2+\delta^2}{(|\bm{R}|^2+\delta^2)^{5/2}} - \frac{\epsilon^2}{|\bm{R}|^5} \bigg) R_2 \, \bm{f}(s-\bars) \, d\bars} &\le \frac{C\delta^2\epsilon}{\delta^2+\epsilon^2}\bigg(\frac{\epsilon^2}{\delta^2+\epsilon^2} +1 \bigg)\norm{\bm{f}}_{C^1(\T)}.
\end{align*}
Lastly, subtracting the estimate \eqref{lemma3_5} with $\delta=0$ from \eqref{lemma3_5} with $\delta>0$, we obtain 
\begin{align*}
\abs{-\frac{3}{2}\int_{-1/2}^{1/2} \bigg(\frac{\epsilon^2+\delta^2}{(|\bm{R}|^2+\delta^2)^{5/2}} - \frac{\epsilon^2}{|\bm{R}|^5} \bigg) R_1 \, \bm{f}(s-\bars) \, d\bars - \frac{2\delta^2}{\epsilon^2+\delta^2}\be_\rho(\be_\rho\cdot\bm{f}(s)) } &\le C(\epsilon^2+\delta^2)^{1/2}\norm{\bm{f}}_{C^1(\T)}.
\end{align*}
Putting the above three bounds together, we obtain the following estimate for $D_2$:
\begin{align*}
\abs{D_2- \frac{2\delta^2}{\epsilon^2+\delta^2}\be_\rho(\be_\rho\cdot\bm{f}(s))} &\le C(\epsilon^2+\delta^2)^{1/2}\norm{\bm{f}}_{C^1(\T)}.
\end{align*}
We can see that $D_2$ serves to cancel the leading order $\be_\rho\be_\rho^{\rm T}$ term of $S_2$, but does not cancel the logarithmic term. \\

Finally, we may estimate $D_3$ using equation \eqref{lemma3_5} of Lemma \ref{CPAM_lemmas} with $m=0$ and $n=5$ to yield
\begin{align*}
\abs{D_3 + \frac{2\delta^2}{\epsilon^2+\delta^2}\bm{f}(s)} &\le \frac{C\delta^2}{(\delta^2+\epsilon^2)^{1/2}} \norm{\bm{f}}_{C^1(\T)}.
\end{align*}
Note that this term exactly cancels the $O(1)$ contribution from the regularized Stokeslet term $S_3$. \\

Combining the above estimates for the three Stokeslet terms and the three doublet terms, we have, for $\bx\in\Gamma_\epsilon$,
\begin{align*}
 &\abs{8\pi\big(\bu^\delta(\bx) - \bu^\SB(\bx)\big) +\log\bigg(1+\frac{\delta^2}{\epsilon^2}\bigg)({\bf I}+\be_{\rm s}\be_{\rm s}^{\rm T}) \bm{f}(s)} \\
 &\hspace{2cm}\le \abs{S_1+\log\bigg(1+\frac{\delta^2}{\epsilon^2}\bigg)\bm{f}(s)} +\abs{S_2 + \log\bigg(1+\frac{\delta^2}{\epsilon^2}\bigg)\be_{\rm s}(\be_{\rm s}\cdot\bm{f}(s)) + 2\frac{\delta^2}{\epsilon^2+\delta^2}\be_\rho(\be_\rho\cdot\bm{f}(s))} \\
 &\hspace{3cm}+ \abs{S_3 - \frac{2\delta^2}{\epsilon^2+\delta^2}\bm{f}(s)} + \abs{D_1} +\abs{D_2- 2\frac{\delta^2}{\epsilon^2+\delta^2}\be_\rho(\be_\rho\cdot\bm{f}(s))} +\abs{D_3 + \frac{2\delta^2}{\delta^2+\epsilon^2}\bm{f}(s)} \\
 &\hspace{2cm}\le C\sqrt{\epsilon^2+\delta^2}\log(\epsilon^2+\delta^2)\norm{\bm{f}}_{C^1(\T)}.
 \end{align*}


To show the centerline estimate \eqref{centerline_est} of Theorem \ref{thm:powerlaw}, we compare $\bu^\delta(\X(s))$, evaluated along the fiber centerline, to $\bu^\SB_{\rm C}(\X(s))$, the centerline expression arising from classical SBT. To prove \eqref{centerline_est}, it will be most convenient to use the expression \eqref{my_SBT} with $\eta=1$, which is shown in \cite{norway} to satisfy the estimate \eqref{center_est}; i.e. the centerline expression \eqref{my_SBT} asymptotically approaches the classical SBT expression \eqref{classicalSBT} evaluated at the slender body surface to $O(\epsilon\log\epsilon)$. \\

Along the fiber centerline, the regularized Stokeslet expression $\bu^\delta$ is given by \eqref{MRS_SBT} with the doublet coefficient $\beta=\frac{\delta^2}{2}$ \eqref{doub_coeff}. Again, this choice of doublet coefficient is due to the need to cancel certain $O(1)$ terms from the regularized Stokeslet (see Remark \ref{rem:doublet}).  \\

Using these expressions, we may write out the fiber velocity difference $\bu^\delta(\X(s)) - \bu^\SB_{\rm C}(\X(s))$ as 
\begin{align*}
\bu^\delta(\X(s)) - \bu^\SB_{\rm C}(\X(s)) &= \frac{1}{8\pi}\bigg(I_1+ I_2 + I_3 + J_1 + J_2 +J_3\bigg), \\
I_1 &= \int_{-1/2}^{1/2} \bigg(\frac{1}{(|\overline{\X}|^2+\delta^2)^{1/2}} - \frac{1}{(|\overline{\X}|^2+\epsilon^2)^{1/2}} \bigg) \bm{f}(s-\bars)d\bars \\
I_2 &= \int_{-1/2}^{1/2}\bigg(\frac{\overline{\X}\overline{\X}^{\rm T}}{(|\overline{\X}|^2+\delta^2)^{3/2}} - \frac{\overline{\X}\overline{\X}^{\rm T}}{(|\overline{\X}|^2+\epsilon^2)^{3/2}} \bigg) \bm{f}(s-\bars)d\bars \\
I_3 &= \int_{-1/2}^{1/2}\frac{\delta^2}{(|\overline{\X}|^2+\delta^2)^{3/2}} \bm{f}(s-\bars)d\bars \\
J_1 &= \frac{1}{2}\int_{-1/2}^{1/2} \bigg(\frac{\delta^2}{(|\overline{\X}|^2+\delta^2)^{3/2}} - \frac{\epsilon^2}{(|\overline{\X}|^2+\epsilon^2)^{3/2}} \bigg) \bm{f}(s-\bars)d\bars \\
J_2 &= -\frac{3}{2}\int_{-1/2}^{1/2}\bigg(\frac{\delta^2\overline{\X}\overline{\X}^{\rm T}}{(|\overline{\X}|^2+\delta^2)^{5/2}} - \frac{\epsilon^2\overline{\X}\overline{\X}^{\rm T}}{(|\overline{\X}|^2+\epsilon^2)^{5/2}} \bigg) \bm{f}(s-\bars)d\bars \\
J_3 &= -\frac{3}{2}\int_{-1/2}^{1/2}\frac{\delta^4}{(|\overline{\X}|^2+\delta^2)^{5/2}} \bm{f}(s-\bars)d\bars.
\end{align*}

Each of the above quantities can be estimated using a straightforward adaptation of Lemma \ref{CPAM_lemmas}. In particular, replacing $\bm{R}=\overline{\X}+\epsilon\be_\rho$ with $\overline{\X} = \X(s) - \X(s-\bars)$ in each denominator, we obtain analogous bounds to estimates \eqref{lemma3_3}--\eqref{lemma3_6} which depend on $\delta$ but not $\epsilon$. Note that these estimates can also be applied for $\delta=\epsilon$, in which case the resulting bound solely depends on $\epsilon$ (which can be used to estimate terms coming from $\bu^\SB_{\rm C}$). Then the above quantities may be estimated in an analogous manner to the previous surface velocity estimates, yielding 
\begin{align*}
\abs{I_1 + \log(\delta^2/\epsilon^2)\bm{f}(s)} &\le C(\epsilon\abs{\log\epsilon}+\delta\abs{\log\delta})\norm{\bm{f}}_{C^1(\T)}\\
\abs{I_2 + \log(\delta^2/\epsilon^2)\be_{\rm s}(\be_{\rm s}\cdot\bm{f}(s))} &\le C(\epsilon\abs{\log\epsilon}+\delta\abs{\log\delta})\norm{\bm{f}}_{C^1(\T)} \\
\abs{I_3 - 2\bm{f}(s)} &\le C(\epsilon+\delta)\norm{\bm{f}}_{C^1(\T)} \\
\abs{J_1 } &\le C(\epsilon+\delta)\norm{\bm{f}}_{C^1(\T)}\\
\abs{J_2 } &\le C(\epsilon+\delta)\norm{\bm{f}}_{C^1(\T)} \\
\abs{J_3 + 2\bm{f}(s)} &\le C(\epsilon+\delta)\norm{\bm{f}}_{C^1(\T)}.
\end{align*} 

Combining the six estimates above, we obtain
\begin{align*}
\abs{8\pi\big(\bu^\delta(\X(s)) - \bu^\SB_{\rm C}(\X(s)) \big) + 2\log\bigg(\frac{\delta}{\epsilon}\bigg)({\bf I}+\be_{\rm s}\be_{\rm s}^{\rm T})\bm{f}(s) } &\le C(\epsilon\abs{\log\epsilon}+\delta\abs{\log\delta})\norm{\bm{f}}_{C^1(\T)}.
\end{align*}

Finally, we show the velocity decay estimate \eqref{decay_est} for fluid points $\bx$ away from the surface of the 3D filament. For simplicity we consider the velocity difference $\bu^\delta-\bu^\SB$ only at points $\bx$ which are close enough to the fiber surface to have a unique representation in the fiber frame \eqref{frameODE}. In particular, estimate \eqref{decay_est} is meant to be a statement about velocity differences close to the fiber surface, rather than far-field agreement. \\

Writing $\bx$ with respect to the fiber frame \eqref{frameODE} as $\bx = \X(s)+\epsilon^{1-\alpha}\be_\rho(s,\theta)$, we define
\begin{align*}
\bm{R}_\alpha = \bx - \X(s-\bars) = \overline{\X}(s,\bars) + \epsilon^{1-\alpha}\be_\rho(s,\theta).
\end{align*}

We then consider the difference $\bu^\delta(\bx)- \bu^\SB(\bx)$ as 
\begin{align*}
\bu^\delta(\bx) &- \bu^\SB(\bx) = \frac{1}{8\pi}\bigg(S_{\alpha,1} + S_{\alpha,2} + D_{\alpha,1} + D_{\alpha,2} \bigg); \\
S_{\alpha,1} &= \int_{-1/2}^{1/2} \bigg(\frac{{\bf I}}{(|\bm{R}_\alpha|^2+\delta^2)^{1/2}}-\frac{{\bf I}}{|\bm{R}_\alpha|} + \frac{\bm{R}_\alpha\bm{R}_\alpha^{\rm T}}{(|\bm{R}_\alpha|^2+\delta^2)^{3/2}} - \frac{\bm{R}_\alpha\bm{R}_\alpha^{\rm T}}{|\bm{R}_\alpha|^3} \bigg)\bm{f}(s-\bars) \, d\bars \\
S_{\alpha,2} &= \int_{-1/2}^{1/2} \frac{\delta^2}{(|\bm{R}_\alpha|^2+\delta^2)^{3/2}} \bm{f}(s-\bars) \, d\bars \\
D_{\alpha,1} &= \frac{\epsilon^2}{2}\int_{-1/2}^{1/2} \bigg(\frac{{\bf I}}{(|\bm{R}_\alpha|^2+\delta^2)^{3/2}}-\frac{{\bf I}}{|\bm{R}_\alpha|^3} -3\bigg(\frac{ \bm{R}_\alpha\bm{R}_\alpha^{\rm T}}{(|\bm{R}_\alpha|^2+\delta^2)^{5/2}} - \frac{\bm{R}_\alpha\bm{R}_\alpha^{\rm T}}{|\bm{R}_\alpha|^5} \bigg) \bigg)\bm{f}(s-\bars) \, d\bars \\
D_{\alpha,2} &= \frac{\delta^2}{2}\int_{-1/2}^{1/2} \bigg(\frac{{\bf I}}{(|\bm{R}_\alpha|^2+\delta^2)^{3/2}}-\frac{3\bm{R}_\alpha\bm{R}_\alpha^{\rm T}}{(|\bm{R}_\alpha|^2+\delta^2)^{5/2}}  - \frac{3(\epsilon^2+\delta^2){\bf I}}{(|\bm{R}_\alpha|^2+\delta^2)^{5/2}} \bigg) \bm{f}(s-\bars) \, d\bars .
\end{align*}

We proceed to estimate each term in turn. 
First note that, in analogy with estimate \eqref{Rdiff}, we have
\begin{equation}\label{Ralpha_diff}
\begin{aligned}
\abs{\frac{1}{(|\bm{R}_\alpha|^2+\delta^2)^{n/2}} - \frac{1}{|\bm{R}_\alpha|^n}}
&\le \frac{\delta^2}{\abs{\bm{R}_\alpha}^{n+2}}.
\end{aligned}
\end{equation}

Using \eqref{Ralpha_diff}, we have
\begin{align*}
\abs{S_{\alpha,1}} &\le 2\norm{\bm{f}}_{C^0(\T)}\int_{-1/2}^{1/2} \frac{\delta^2}{|\bm{R}_\alpha|^3} \, d\bars \le \frac{C\delta^2}{\epsilon^{2(1-\alpha)}+\delta^2}\norm{\bm{f}}_{C^0(\T)} \\
&\le \frac{C\delta^2}{\epsilon^{2(1-\alpha)}}\norm{\bm{f}}_{C^0(\T)},
\end{align*}
where we have used Lemma \ref{CPAM_lemmas}, estimate \eqref{lemma3_3} with $\epsilon^{1-\alpha}$ in place of $\epsilon$. \\

Similarly, again using \eqref{Ralpha_diff} along with estimate \eqref{lemma3_3} (again replacing $\epsilon$ with $\epsilon^{1-\alpha}$), we have
\begin{align*}
\abs{D_{\alpha,1}} &\le 2\epsilon^2\norm{\bm{f}}_{C^0(\T)}\int_{-1/2}^{1/2} \frac{\delta^2}{|\bm{R}_\alpha|^5} \, d\bars \le \frac{C\delta^2\epsilon^2}{(\epsilon^{2(1-\alpha)}+\delta^2)^2} \norm{\bm{f}}_{C^0(\T)} \\
&\le \frac{C\delta^2}{\epsilon^{2(1-2\alpha)}}\norm{\bm{f}}_{C^0(\T)}.
\end{align*}

Finally, using Lemma \ref{CPAM_lemmas}, estimate \eqref{lemma3_3} with $\epsilon^{1-\alpha}$ in place of $\epsilon$, we may bound $S_{\alpha,2}$ and $D_{\alpha,2}$ as
\begin{align*}
\abs{S_{\alpha,2}} &\le \frac{C\delta^2}{\epsilon^{2(1-\alpha)}+\delta^2}\le \frac{C\delta^2}{\epsilon^{2(1-\alpha)}}\norm{\bm{f}}_{C^0(\T)}
\end{align*}
and
\begin{align*}
\abs{D_{\alpha,2}} &\le \norm{\bm{f}}_{C^0(\T)}\int_{-1/2}^{1/2} \bigg(\frac{2\delta^2}{(|\bm{R}_\alpha|^2+\delta^2)^{3/2}}  + \frac{3\delta^2(\epsilon^2+\delta^2)}{2(|\bm{R}_\alpha|^2+\delta^2)^{5/2}} \bigg)  \, d\bars \\
&\le \frac{C\delta^2}{\epsilon^{2(1-\alpha)}+\delta^2}\bigg(1+\frac{\epsilon^2+\delta^2}{\epsilon^{2(1-\alpha)}+\delta^2}\bigg)\norm{\bm{f}}_{C^0(\T)} 
\le \frac{C\delta^2}{\epsilon^{2(1-\alpha)}}\norm{\bm{f}}_{C^0(\T)}. 
\end{align*}

Combining the above four estimates, we obtain the decay bound \eqref{decay_est} near the slender body surface.

\subsection{Force along fiber surface}\label{subsec:force}

To show the force estimate \eqref{force_est} from Theorem \ref{thm:powerlaw}, we first call attention to the two terms appearing in the $\bu^\delta$ expression \eqref{MRS_SBT} of the form
\begin{align*}
\bm{d}_n^\delta[\bm{f}](\bx) = \frac{1}{8\pi}\int_{\T}\frac{\delta^2\bm{f}(s')}{(|\bm{R}|^2+\delta^2)^{n/2}} ds',
\end{align*}
one due to the regularized Stokeslet and the other due to the regularized doublet. In addition, we focus on the term from the regularized pressure \eqref{pressure_reg} of the form 
\begin{align*}
 \bm{d}^\delta_5[\bm{f}\cdot\bm{R}] = \frac{1}{8\pi}\int_{\T}\frac{\delta^2\bm{R}\cdot\bm{f}(s')}{(|\bm{R}|^2+\delta^2)^{5/2}} ds'.
 \end{align*} 

The contribution of these three terms to the force density $\bm{f}^\delta(s)$ \eqref{reg_force} along $\Gamma_\epsilon$ is given by
\begin{equation}\label{Fdelta}
\bm{F}^\delta(s) :=  \int_0^{2\pi}\bigg(2\E(\bm{d}_3^\delta[\bm{f}])\bm{n} - 3(\epsilon^2+\delta^2)\E(\bm{d}_5^\delta[\bm{f}])\bm{n} - 3 (\bm{d}^\delta_5[\bm{f}\cdot\bm{R}])\bm{n} \bigg)\mc{J}_\epsilon(s,\theta)d\theta.
\end{equation}

Note that besides the three $\bm{d}^\delta$ terms, each of the other terms contributing to the regularized force definition \eqref{reg_force} have direct non-regularized analogues along the fiber surface. Thus each term of $\bm{f}^\delta-\bm{F}^\delta$ may be estimated following the exact same procedure as in \cite{closed_loop}, Section 3.3, but using the $\delta$-dependent estimates of Lemma \ref{CPAM_lemmas} in place of Lemmas 3.3--3.5 in \cite{closed_loop}. In doing this, we obtain the following bound on $\bm{f}^\delta-\bm{F}^\delta$:
\begin{equation}\label{old_force}
\abs{(\bm{f}^\delta-\bm{F}^\delta) - \frac{\bm{f}}{1+\delta^2/\epsilon^2}} \le C\sqrt{\epsilon^2+\delta^2}\norm{\bm{f}}_{C^1(\T)}.
\end{equation}

Therefore to show the force estimate \eqref{force_est} of Theorem \ref{thm:powerlaw}, it remains to calculate the additional force contribution from $\bm{F}^\delta$. 

\begin{lemma}\label{lem:Fdelta}
Given a force density $\bm{f}\in C^1(\T)$, let $\bm{F}^\delta$ be given by \eqref{Fdelta}. Then, for $\epsilon$, $\delta$ sufficiently small we have
\begin{equation}
\abs{\bm{F}^\delta(s) + \frac{\delta^2/\epsilon^2}{(1+ \delta^2/\epsilon^2)^2}\bm{f}(s)} \le C\sqrt{\epsilon^2+\delta^2}\norm{\bm{f}}_{C^1(\T)}.
\end{equation}
\end{lemma}

\begin{proof}
Using that $\bm{n}(s,\theta)=-\be_\rho(s,\theta)$, we begin by noting that
\begin{align*}
&2\E(\bm{d}_n^\delta[\bm{f}])\bm{n}\\
 &= \int_{\T}\frac{n\delta^2}{(|\bm{R}|^2+\delta^2)^{(n+2)/2}}\bigg((\bm{R}\cdot\be_\rho)\bm{f}(s') + (\bm{f}(s')\cdot\be_\rho)\big((\bm{R}\cdot\be_\rho)\be_\rho + (\bm{R}\cdot\be_\theta)\be_\theta + (\bm{R}\cdot\be_{\rm s})\be_{\rm s} \big) \bigg) \, ds'.
\end{align*}

Using the definition \eqref{Fdelta} of $\bm{F}^\delta$, we may write  
\begin{align*}
\bm{F}^\delta &= F_1 + F_2 + F_3 + F_4 + F_5 + F_6; \\
 F_1&= \int_0^{2\pi}\int_{-1/2}^{1/2} \bigg(2\E(\bm{d}_3^\delta[\bm{f}])\bm{n} - 3(\epsilon^2+\delta^2)\E(\bm{d}_5^\delta[\bm{f}])\bm{n} - 3(\bm{d}^\delta_5[\bm{f}\cdot\bm{R}])\bm{n}\bigg)\, d\bars \, (-\epsilon^2\wh{\kappa})d\theta \\
  F_2&= \frac{1}{8\pi}\int_0^{2\pi}\int_{-1/2}^{1/2}\frac{3\delta^2}{(|\bm{R}|^2+\delta^2)^{5/2}}\bigg( (\bm{R}\cdot\be_\theta)\be_\theta +(\bm{R}\cdot\be_{\rm s})\be_{\rm s} \bigg)(\bm{f}(s-\bars)\cdot\be_\rho) \, d\bars \; \epsilon d\theta \\
 F_3&= \frac{1}{8\pi}\int_0^{2\pi}\int_{-1/2}^{1/2}\frac{3\delta^2}{(|\bm{R}|^2+\delta^2)^{5/2}} (\bm{R}\cdot\be_\rho)\bigg(\bm{f}(s-\bars) +  (\bm{f}(s-\bars)\cdot\be_\rho)\be_\rho \bigg) \, d\bars \; \epsilon d\theta \\
 F_4&= -\frac{15}{16\pi}\int_0^{2\pi}\int_{-1/2}^{1/2}\frac{\delta^2(\epsilon^2+\delta^2)}{(|\bm{R}|^2+\delta^2)^{7/2}} \bigg((\bm{R}\cdot\be_\theta)\be_\theta + (\bm{R}\cdot\be_{\rm s})\be_{\rm s}\bigg)(\bm{f}(s-\bars)\cdot\be_\rho) \, d\bars \; \epsilon d\theta \\
 F_5&= -\frac{15}{16\pi}\int_0^{2\pi}\int_{-1/2}^{1/2}\frac{\delta^2(\epsilon^2+\delta^2)}{(|\bm{R}|^2+\delta^2)^{7/2}}(\bm{R}\cdot\be_\rho)\bigg(\bm{f}(s-\bars) +(\bm{f}(s-\bars)\cdot\be_\rho)\be_\rho \bigg) \, d\bars \; \epsilon d\theta \\
 F_6 &= \frac{1}{8\pi}\int_0^{2\pi}\int_{-1/2}^{1/2}\frac{3\delta^2}{(|\bm{R}|^2+\delta^2)^{5/2}}\big(\bm{R}\cdot\bm{f}(s-\bars) \big) \be_\rho \, d\bars \; \epsilon d\theta.
 \end{align*} 
As with the previous velocity bounds, we estimate each of the above terms individually. The first term $F_1$ arises due to the nonzero curvature \eqref{kappahat} of the fiber. 
By Lemma \ref{CPAM_lemmas}, estimate \eqref{lemma3_3}, we have that
 \begin{align*}
 \abs{F_1} &\le C\norm{\bm{f}}_{C^0(\T)}\int_{-1/2}^{1/2}\bigg(\frac{\delta^2\epsilon^2}{(|\bm{R}|^2+\delta^2)^2} + \frac{\delta^2\epsilon^2(\delta^2+\epsilon^2)}{(|\bm{R}|^2+\delta^2)^3} \bigg)d\bars \le \frac{C\epsilon^2\delta^2}{(\epsilon^2+\delta^2)^{3/2}}\norm{\bm{f}}_{C^0(\T)} .
 \end{align*}

To bound the remaining terms, we will need that 
\begin{equation}\label{Rdots}
\bm{R}\cdot\be_\rho = \bars^2(\bm{Q}\cdot\be_\rho)+\epsilon, \quad \bm{R}\cdot\be_\theta = \bars^2(\bm{Q}\cdot\be_\theta), \quad \bm{R}\cdot\be_{\rm s} = \bars + \bars^2(\bm{Q}\cdot\be_{\rm s}).
\end{equation}

Then, using Lemma \ref{CPAM_lemmas}, estimates \eqref{lemma3_4} and \eqref{lemma3_3}, we have
\begin{align*}
\abs{F_2} &= \abs{\frac{1}{8\pi}\int_0^{2\pi}\int_{-1/2}^{1/2}\frac{3\delta^2\epsilon}{(|\bm{R}|^2+\delta^2)^{5/2}}\bigg(\bars\be_{\rm s} + \bars^2\big((\bm{Q}\cdot\be_\theta)\be_\theta +(\bm{Q}\cdot\be_{\rm s}))\be_{\rm s}\big) \bigg)(\bm{f}(s-\bars)\cdot\be_\rho) \, d\bars \, d\theta} \\
&\le \frac{C\delta^2\epsilon}{\delta^2+\epsilon^2}\norm{\bm{f}}_{C^1(\T)}.
\end{align*}

Similarly, the term $F_4$ satisfies
\begin{align*}
\abs{F_4} &= \abs{\frac{15}{16\pi}\int_0^{2\pi}\int_{-1/2}^{1/2}\frac{(\epsilon^2+\delta^2)\delta^2\epsilon}{(|\bm{R}|^2+\delta^2)^{7/2}}\bigg(\bars\be_{\rm s} + \bars^2\big((\bm{Q}\cdot\be_\theta)\be_\theta +(\bm{Q}\cdot\be_{\rm s}))\be_{\rm s}\big) \bigg)(\bm{f}(s-\bars)\cdot\be_\rho) \, d\bars \, d\theta} \\
&\le \frac{C\delta^2\epsilon}{\delta^2+\epsilon^2}\norm{\bm{f}}_{C^1(\T)}.
\end{align*}

To estimate $F_3$ and $F_5$, we will need to use that Lemma \ref{CPAM_lemmas}, estimate \eqref{lemma3_5} implies that for any $\bm{g}\in C^1(\Gamma_\epsilon)$, along the surface of the fiber, we have 
\begin{equation}\label{theta_int}
\begin{aligned}
&\abs{\int_0^{2\pi}\int_{-1/2}^{1/2}\frac{\bars^m}{(|\bm{R}|^2+\delta^2)^{n/2}}\bm{g}(\bars,\theta)\, d\bars d\theta - (\epsilon^2+\delta^2)^{(m+1-n)/2}d_{mn}\int_0^{2\pi}\bm{g}(0,\theta)d\theta} \\
&\hspace{3cm}\le C \max_{0\le\theta<2\pi}\norm{\bm{g}(\cdot,\theta)}_{C^1(\T)}(\epsilon^2+\delta^2)^{(m+2-n)/2},
\end{aligned}
\end{equation}
where $m,n$ and $d_{mn}$ are as in equation \eqref{lemma3_5}. Note furthermore that, using \eqref{erho}, 
\begin{align*}
\int_0^{2\pi} \be_\rho(s,\theta)\be_\rho(s,\theta)^{\rm T} \,d\theta &= \pi (\be_{n_1}\be_{n_1}^{\rm T}+\be_{n_2}\be_{n_2}^{\rm T})= \pi ({\bf I} - \be_{\rm s}\be_{\rm s}^{\rm T})
\end{align*}

Then, using \eqref{Rdots}, \eqref{lemma3_3}, and estimate \eqref{theta_int}, $F_3$ satisfies 
\begin{align*}
\abs{F_3 - \frac{\delta^2\epsilon^2}{2(\delta^2+\epsilon^2)^2}(3{\bf I}-\be_{\rm s}\be_{\rm s}^{\rm T})\bm{f}(s)} & \le \frac{C\delta^2\epsilon^2}{(\delta^2+\epsilon^2)^{3/2}}\norm{\bm{f}}_{C^1(\T)},
\end{align*}
while the term $F_5$ satisfies
\begin{align*}
\abs{F_5 + \frac{\delta^2\epsilon^2}{(\delta^2+\epsilon^2)^2}(3{\bf I}-\be_{\rm s}\be_{\rm s}^{\rm T})\bm{f}(s)} &\le \frac{C\delta^2\epsilon^2}{(\delta^2+\epsilon^2)^{3/2}}\norm{\bm{f}}_{C^1(\T)}.
\end{align*}

Finally, using \eqref{Rdef} along with Lemma \ref{CPAM_lemmas}, estimates \eqref{lemma3_3}, \eqref{lemma3_4}, and \eqref{lemma3_5}, we have that the pressure term $F_6$ satisfies 
\begin{align*}
\abs{F_6 - \frac{\delta^2\epsilon^2}{2(\delta^2+\epsilon^2)^2}({\bf I}-\be_{\rm s}\be_{\rm s}^{\rm T})\bm{f}(s)} &= \bigg|\frac{1}{8\pi}\int_0^{2\pi}\int_{-1/2}^{1/2}\frac{3\delta^2\epsilon(\bars\be_{\rm s}+\epsilon\be_\rho+\bars^2\bm{Q})}{(|\bm{R}|^2+\delta^2)^{5/2}} \cdot\bm{f}(s-\bars) \be_\rho \, d\bars \; d\theta \\
&\qquad - \frac{\delta^2\epsilon^2}{2(\delta^2+\epsilon^2)^2}({\bf I}-\be_{\rm s}\be_{\rm s}^{\rm T})\bm{f}(s) \bigg| \\
&\le \bigg|\frac{1}{8\pi}\int_0^{2\pi}\int_{-1/2}^{1/2}\frac{3\delta^2\epsilon^2\be_\rho\be_\rho^{\rm T}}{(|\bm{R}|^2+\delta^2)^{5/2}} \bm{f}(s-\bars) \, d\bars \; d\theta \\
&\qquad - \frac{\delta^2\epsilon^2}{2(\delta^2+\epsilon^2)^2}({\bf I}-\be_{\rm s}\be_{\rm s}^{\rm T})\bm{f}(s) \bigg| + \frac{C\delta^2\epsilon}{\delta^2+\epsilon^2}\norm{\bm{f}}_{C^1(\T)} \\
&\le \frac{C\delta^2\epsilon^2}{(\delta^2+\epsilon^2)^{3/2}}\norm{\bm{f}}_{C^1(\T)}.
\end{align*}

Altogether, the additional force contribution from $\bm{F}^\delta$ along $\Gamma_\epsilon$ satisfies 
\begin{align*}
\abs{\bm{F}^\delta + \frac{\delta^2\epsilon^2}{(\delta^2+\epsilon^2)^2}\bm{f}(s)} &\le \abs{F_1}+\abs{F_2}+\abs{F_3}+\abs{F_4-\frac{\delta^2\epsilon^2}{2(\delta^2+\epsilon^2)^2}(3{\bf I}-\be_{\rm s}\be_{\rm s}^{\rm T})\bm{f}(s)} \\
&\quad + \abs{F_5+\frac{\delta^2\epsilon^2}{(\delta^2+\epsilon^2)^2}(3{\bf I}-\be_{\rm s}\be_{\rm s}^{\rm T})\bm{f}(s)} + \abs{F_6 - \frac{\delta^2\epsilon^2}{2(\delta^2+\epsilon^2)^2}({\bf I}-\be_{\rm s}\be_{\rm s}^{\rm T})\bm{f}(s)} \\
&\le \frac{C\epsilon\delta}{\sqrt{\epsilon^2+\delta^2}}\norm{\bm{f}}_{C^1(\T)},
\end{align*}
from which we obtain Lemma \ref{lem:Fdelta}.

\end{proof}

\section{Numerical comparison}\label{sec:numerics}
To gain a better sense of the practical implications of Theorem \ref{thm:powerlaw}, we turn to numerics to better visualize discrepancies. We aim to address whether Theorem \ref{thm:powerlaw} actually has a noticeable effect in practice.

\subsection{Centerline difference}\label{subsec:num_center}
In Table \ref{tab:centerline} we provide numerical verification of the velocity difference along the fiber centerline predicted by \eqref{centerline_est}. Here we consider a slender torus of length 1 centered at the origin whose centerline lies in the $xy$-plane, and we prescribe a constant force $\bm{f}=\be_x$ along the fiber. We can see that even in this very simple scenario, choosing $\delta=\epsilon$ eliminates a sizeable $O(1)$ discrepancy between regularized and classical SBT. \\
\begin{table}[!ht]
\centering
\begin{tabular}{ | c | c | c | c | c | c |} 
 \hline 
 & \multicolumn{5}{|c|}{ $\norm{\bu^\delta(\X(\cdot))-\bu^\SB_{\rm C}}_{L^\infty(\T)}$ } \\ 
\hline 
$\epsilon$  & $\frac{\delta}{\epsilon}=\frac{1}{3}$ & $\frac{\delta}{\epsilon}=\frac{1}{2}$ & $\frac{\delta}{\epsilon}=1$ & $\frac{\delta}{\epsilon}=2$ & $\frac{\delta}{\epsilon}=3$ \\  
\hline
0.01 & 0.1705 & 0.1093 & 9.6153e-05 & 0.1088 & 0.1715  \\   
0.005 & 0.1708  & 0.1095  & 1.9382e-05 & 0.1099 & 0.1738  \\ 
0.001 & 0.1709 & 0.1096  & 1.3324e-05 & 0.1103 & 0.1748  \\ 
0.0005 & 0.1709  & 0.1096 & 1.4773e-05 & 0.1103  & 0.1748 \\
\hline\hline
$\abs{\frac{1}{2\pi}\log(\delta/\epsilon)}$ & 0.17485 & 0.11032 & 0 & 0.11032 & 0.17485 \\
\hline
\end{tabular}
\vspace{0.2cm}
\caption{$L^\infty$ norm of the difference between fiber velocity expressions $\bu^\delta(\X(s))$ \eqref{MRS_SBT} and $\bu^\SB_{\rm C}(\X(s))$ \eqref{my_SBT}  for different ratios of $\delta/\epsilon$. Note that $\frac{\delta}{\epsilon}=1$ is necessary to avoid an $O(1)$ difference that is much larger than the fiber radius. Here we have prescribed the constant force $\bm{f}=\be_x$ and use $\frac{2}{\epsilon}$ discretization points to fully resolve the behavior of the nearly singular integrals. }
\label{tab:centerline}
\end{table}

\subsection{Bulk difference}\label{subsec:bulk_diff}
In the bulk, where an $O(1)$ discrepancy between regularized and classical SBT cannot be eliminated, we are interested in visualizing flow differences around the slender body and their effects on multiple interacting fibers. For each of the following figures we use the regularization parameter $\delta=\epsilon$, since this results in the best agreement in fiber velocity expressions. \\

We first consider the flow around a single torus of length 1 with centerline in the $xy$-plane, and again consider the constant force $\bm{f}=\be_x$. 
Figure \ref{fig:torus1} shows the resulting difference $\abs{\bu^\delta(\bx)-\bu^\SB(\bx)}$ for fluid points $\bx$ around and up to the surface of the fiber for two different fiber radii ($\epsilon=0.01$ and $\epsilon=0.001$). Note that the magnitude of the maximum discrepancy between regularized and classical SBT remains the same as $\epsilon$ decreases, as stated in Theorem \ref{thm:powerlaw}, but the magnitude of this discrepancy is quite small (since $\delta=\epsilon$, the maximum discrepancy is $\frac{1}{4\pi}\log(2)\approx 0.05516$). We can also see that the discrepancy is largest where the prescribed force $\bm{f}=\be_x$ is tangent to the fiber surface (i.e. along the top and bottom of the fiber), but the discrepancy extends further into the bulk when the force is normal to the fiber.
Away from the surface of the fiber, we can see how the difference between $\bu^\delta(\bx)$ and $\bu^\SB(\bx)$ does decrease with $\epsilon$, as noted in Theorem \ref{thm:powerlaw}. In particular, any major discrepancies are localized very close to the fiber surface for small $\epsilon$.   \\
\begin{figure}[!ht]
\centering
\includegraphics[width=0.488\linewidth]{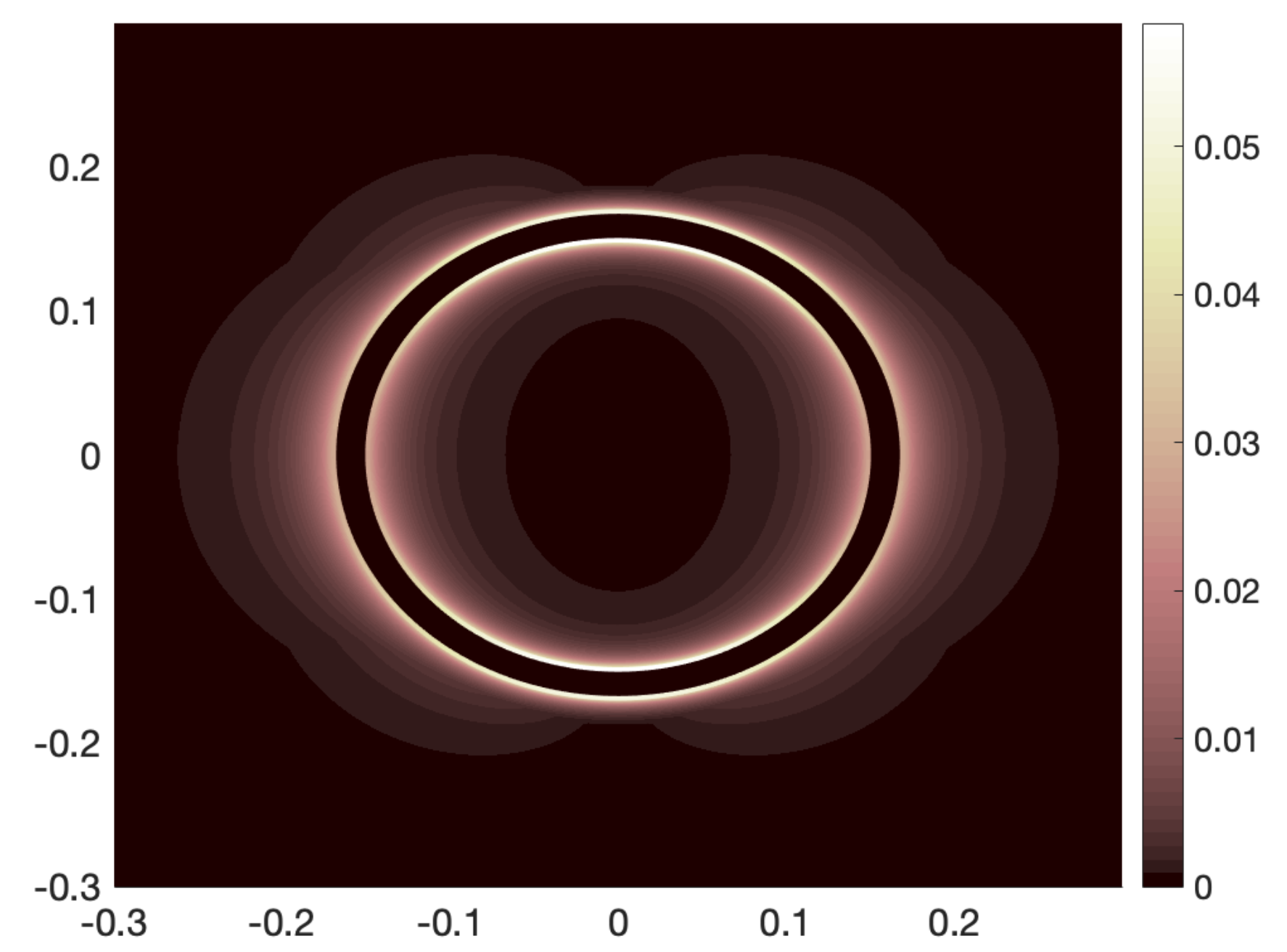} 
\includegraphics[width=0.48\linewidth]{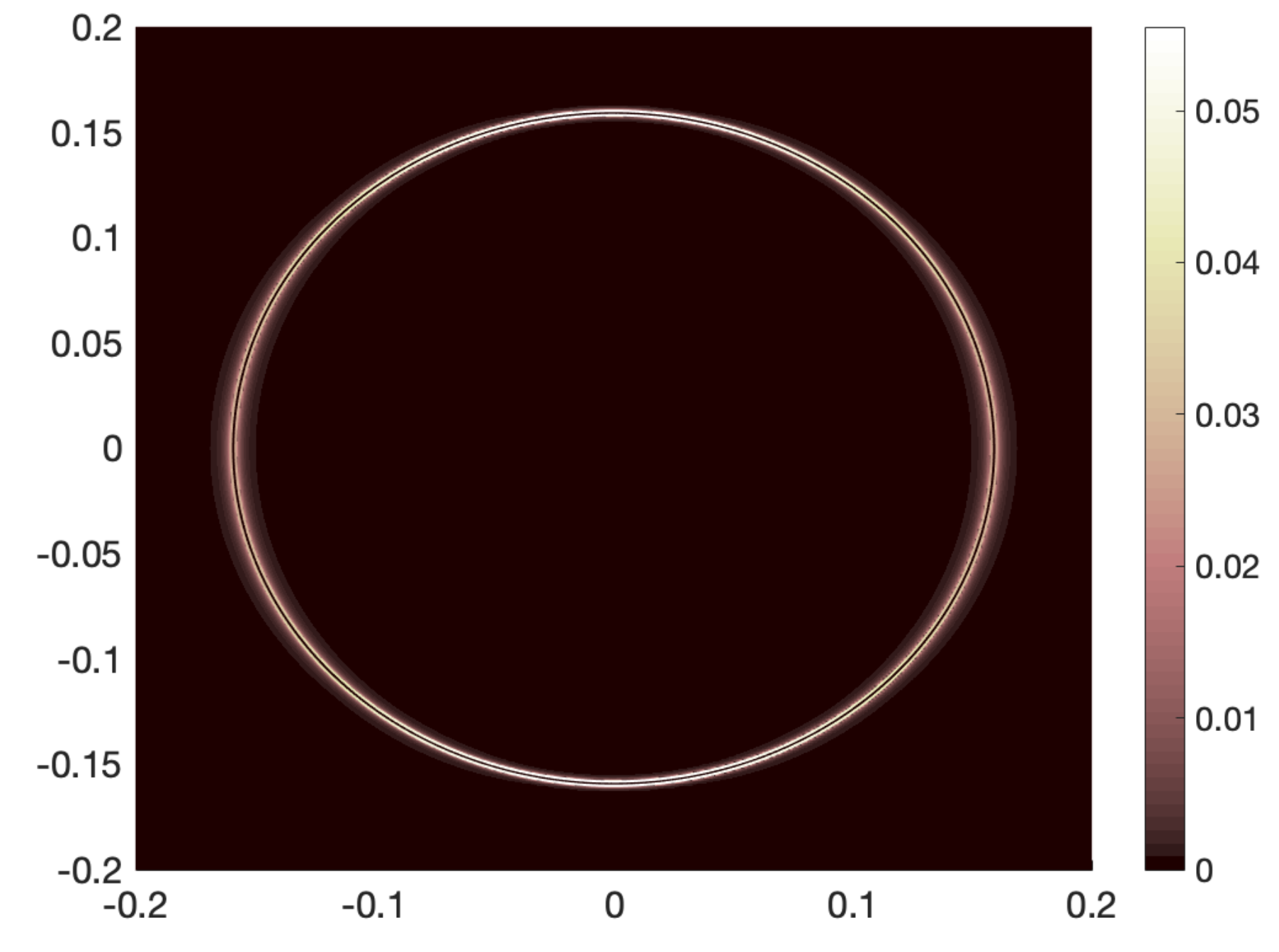} 
 \caption{Plot of the difference $\abs{\bu^\delta(\bx)-\bu^\SB(\bx)}$ for $\bx$ around a thin torus of length 1 with centerline in the $xy$-plane, subject to the forcing $\bm{f}=\be_x$, for $\epsilon=0.01$ (left) and $\epsilon=0.001$ (right). The maximum difference $\frac{1}{4\pi}\log(2)\approx 0.05516$ occurs at the fiber surface along the top and bottom of the fiber, where the force is tangent to the fiber surface. Note, however, that discrepancies extend further into the bulk on the left and right sides of the fiber, where the force is normal to the centerline.}
  \label{fig:torus1}
\end{figure}

We next consider the interaction between two nearby fibers as well as the flow around them. 
Figure \ref{fig:torus2} shows two thin tori of length 1 and radius $\epsilon=0.01$ with centerlines in the $xy$-plane, both centered along the $x$-axis at $x=0$ and $x=0.4$. For the constant force $\bm{f}=\be_x$, we plot the difference $\abs{\bu^\delta(\bx)-\bu^\SB(\bx)}$.\\

\begin{figure}[!ht]
\centering
\includegraphics[width=0.8\linewidth]{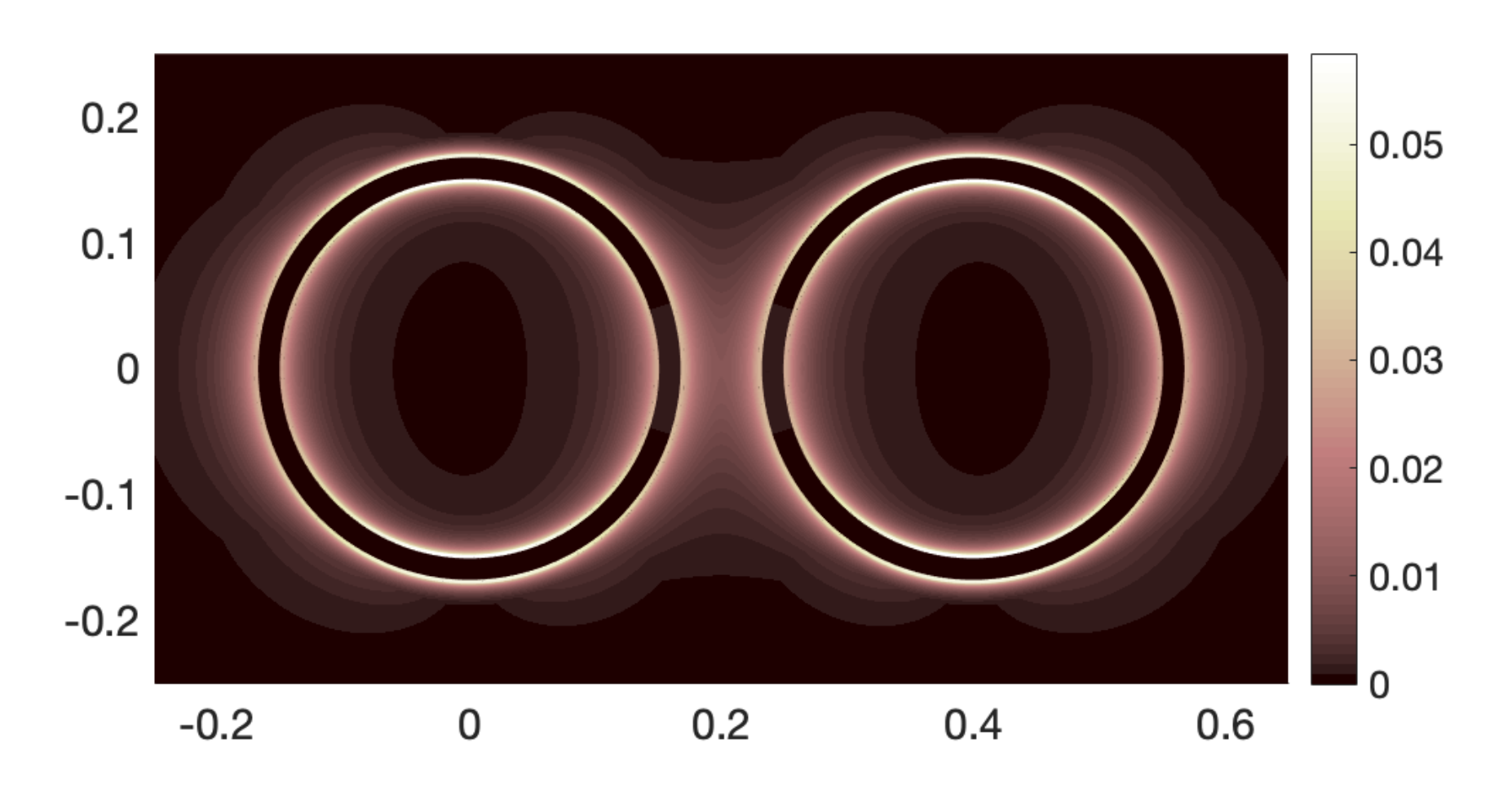} 
 \caption{Plot of the difference $\abs{\bu^\delta(\bx)-\bu^\SB(\bx)}$ for $\bx$ around two thin tori in the $xy$-plane, both subject to the constant force $\bm{f}=\be_x$. Both tori have radius $\epsilon=0.01$. The maximum difference is almost exactly the same as in the single torus case, indicating that the $O(1)$ surface discrepancy between regularized and classical SBT may not noticeably affect the flow between multiple hydrodynamically interacting fibers. }
  \label{fig:torus2}
\end{figure}

For two thin tori interacting hydrodynamically, we also take a closer look at velocity differences between regularized and classical SBT along the fiber itself. We consider first consider two tori of radius $\epsilon=0.01$ with centerlines in the $xy$-plane (see Figure \ref{fig:2torus_only}). The difference in fiber velocities $\bu^\delta(\X(s))$ and $\bu^\SB(\X(s))$ is largest where the fibers are closest to each other, but the magnitude of this difference is still smaller than the fiber radius, even when the fibers are very close together (see Figure \ref{fig:2torus_only}, top right). \\
\begin{figure}[!ht]
\centering
\includegraphics[width=0.49\linewidth]{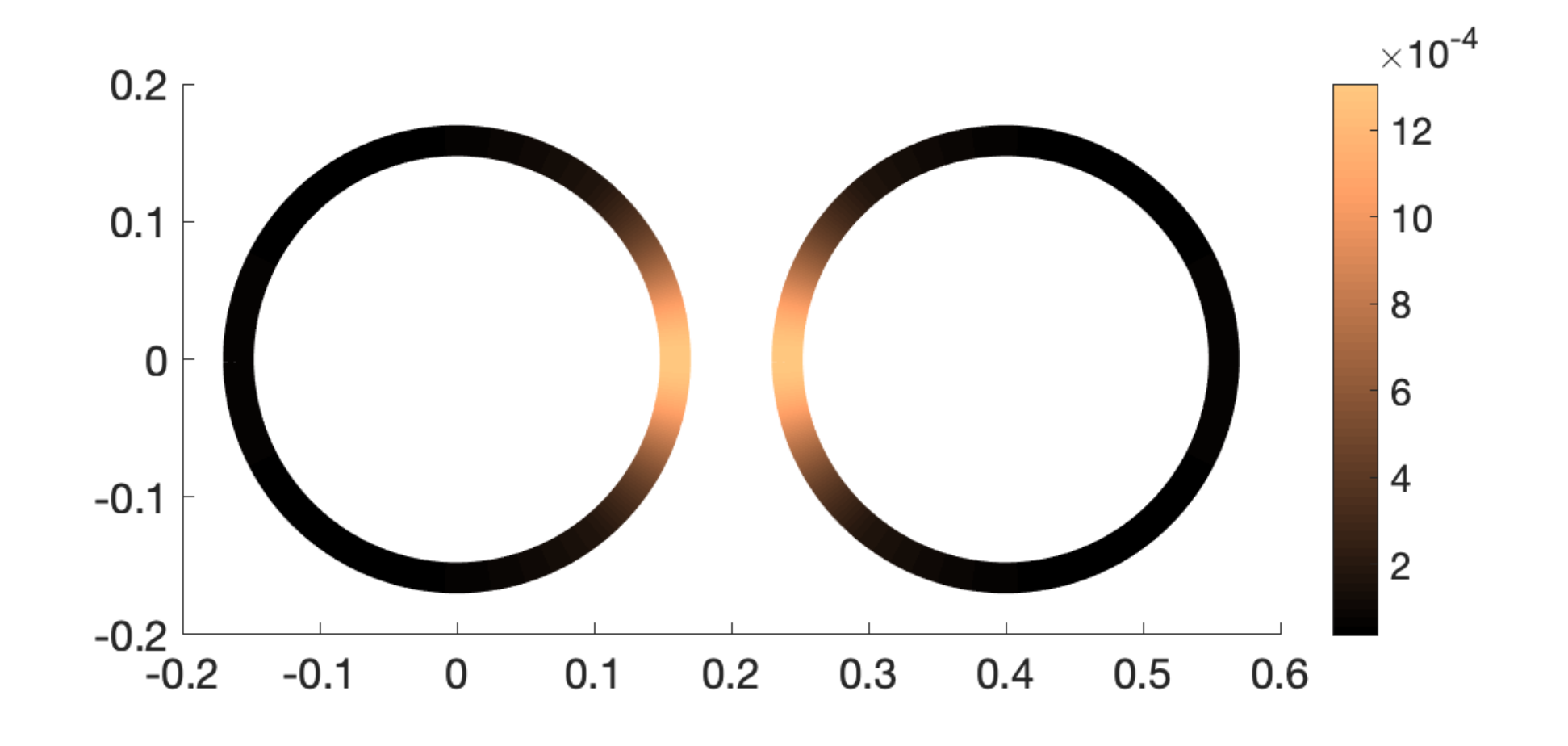} 
\includegraphics[width=0.49\linewidth]{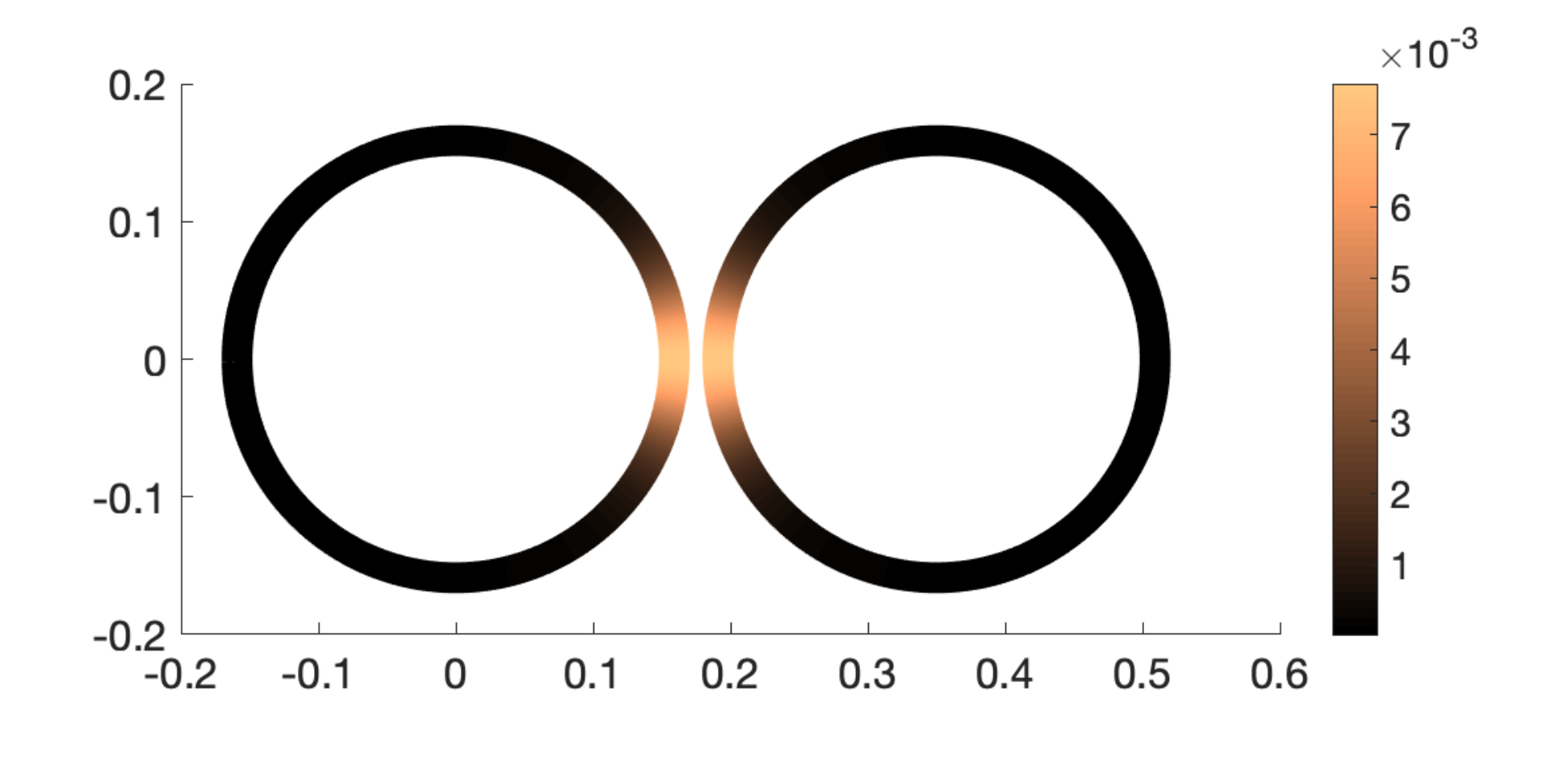} \\
\includegraphics[width=0.49\linewidth]{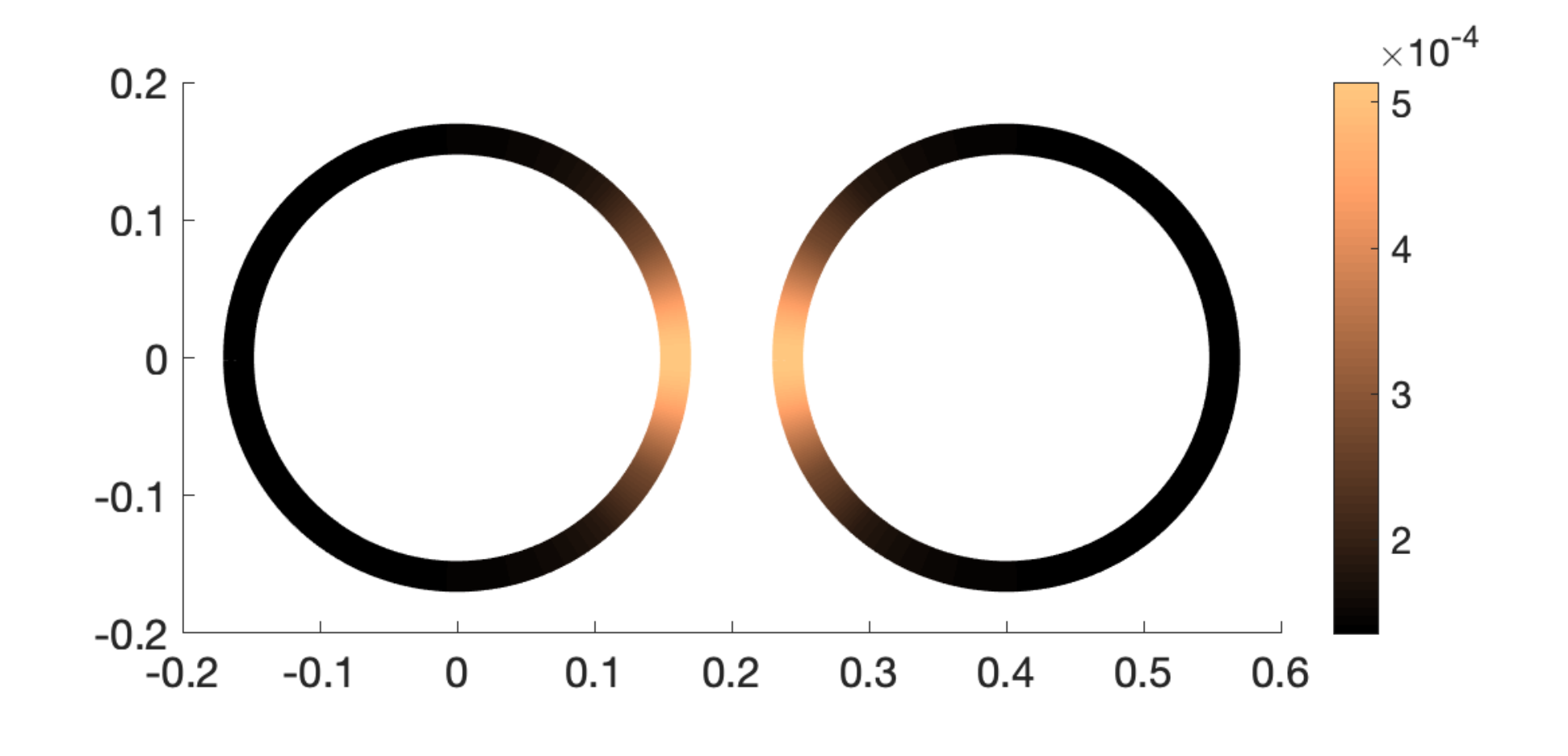} 
\includegraphics[width=0.49\linewidth]{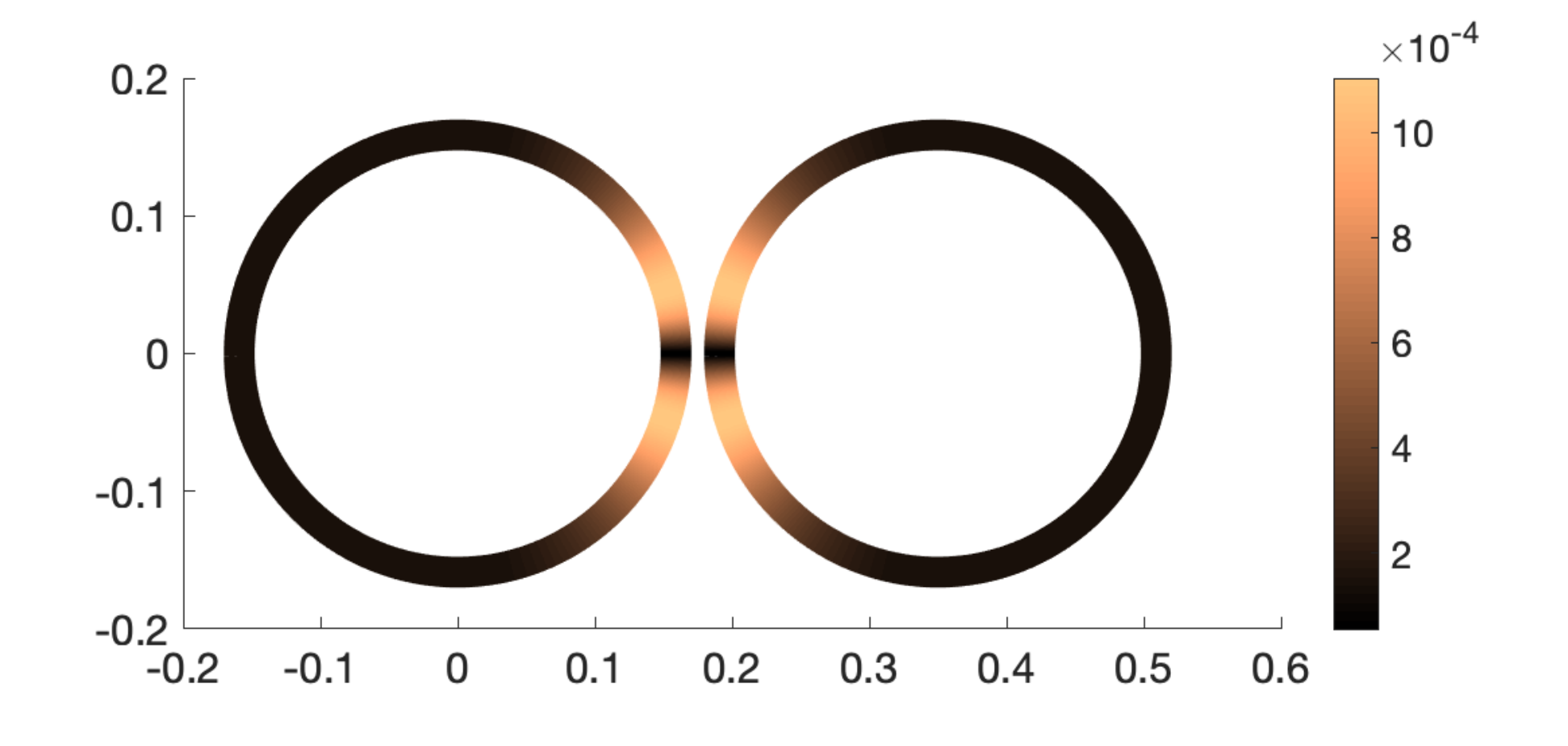} 
 \caption{Plot of the centerline difference $\abs{\bu^\delta(\X(s))-\bu^\SB_{\rm C}(\X(s))}$ for two hydrodynamically interacting thin tori in the $xy$-plane. Both tori have radius $\epsilon=0.01$. The top fibers are subject to the force $\bm{f}=\be_x$, while the bottom fibers are subject to the force $\bm{f}=\be_y$. The fibers on the left are $0.0617$ apart at their closest point, while the fibers on the right are $0.0117$ apart at their closest point. The maximum differences $\norm{\bu^\delta(\X(\cdot))-\bu^\SB_{\rm C}(\X(\cdot))}_{L^\infty(\T)}$ are (clockwise from top right): 0.0013, 0.0077, 0.00051, 0.0011. These differences are small, but significantly larger than when only a single fiber is present (see Table \ref{tab:centerline}, middle column). }
  \label{fig:2torus_only}
\end{figure}

We also consider two tori of radius $\epsilon=0.01$ with $xy$-planar centerlines that are stacked in the $z$-direction (see Figure \ref{fig:2torus_only2}). The fibers are a distance $0.01$ from each other at every point, and subject to a constant force $\bm{f}=\be_x$ (left) and $\bm{f}=\be_z$ (right). Again, the largest difference $\norm{\bu^\delta(\X(\cdot))-\bu^\SB_{\rm C}(\X(\cdot))}_{L^\infty(\T)}=0.0076$ (Figure \ref{fig:2torus_only2}, right) is still smaller than the fiber radius. This difference is, however, significantly larger than when only a single fiber is present (compare to Table \ref{tab:centerline}, center column). \\
\begin{figure}[!ht]
\centering
\includegraphics[width=0.49\linewidth]{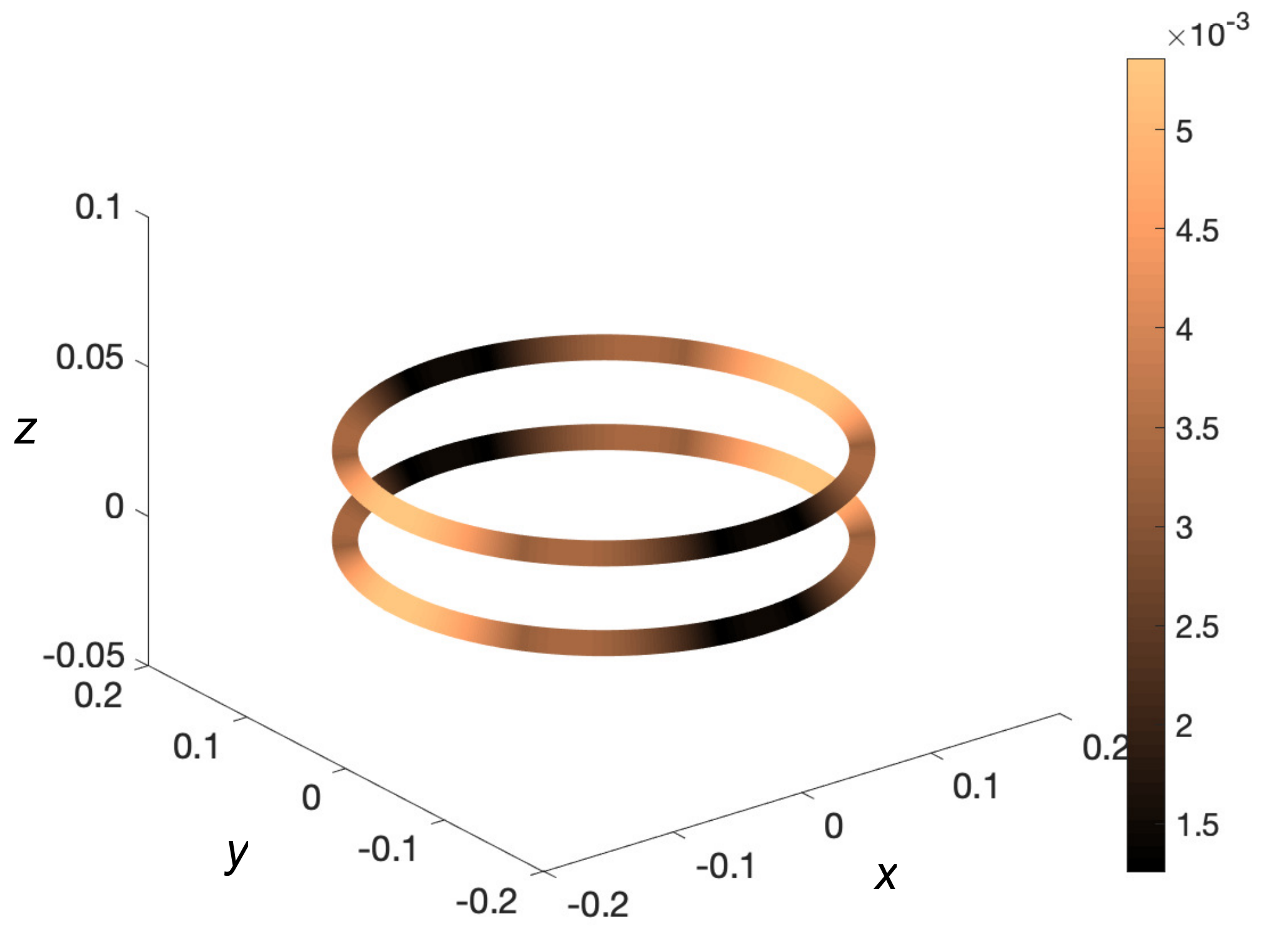} 
\includegraphics[width=0.49\linewidth]{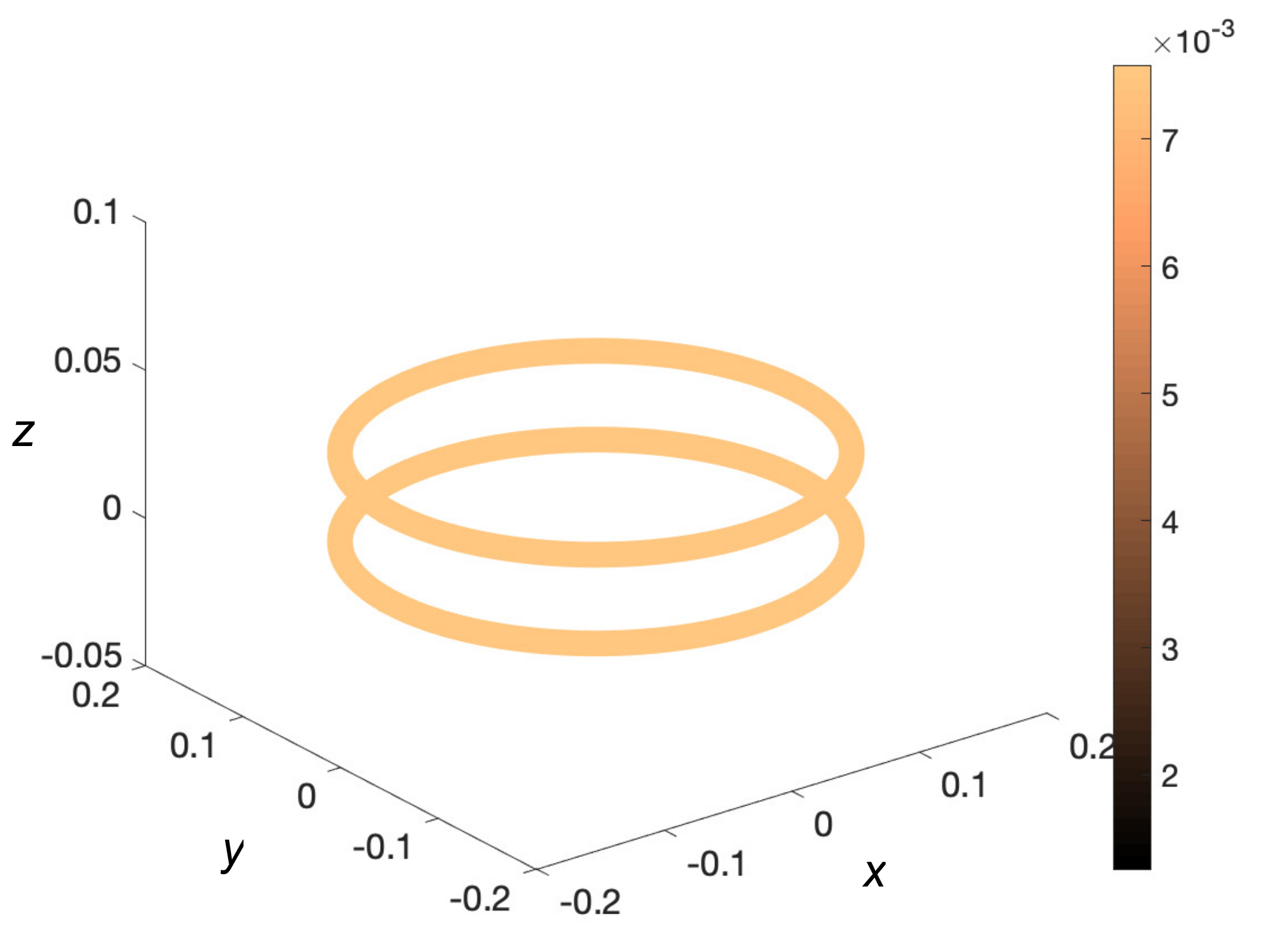} \\ 
 \caption{Plot of the centerline difference $\abs{\bu^\delta(\X(s))-\bu^\SB_{\rm C}(\X(s))}$ for two hydrodynamically interacting thin tori of radius $\epsilon=0.01$, stacked in the $z$-direction at a distance $0.01$ from each other. The left fibers are subject to the constant force $\bm{f}=\be_x$ and have a maximum difference of 0.0054. The fibers on the right are subject to the force 
$\bm{f}=\be_z$ and have a maximum difference of 0.0076.}
  \label{fig:2torus_only2}
\end{figure}

The main takeaway from Figures \ref{fig:torus2} -- \ref{fig:2torus_only2} is that the magnitude of the discrepancy between regularized and classical SBT noted in Theorem \ref{thm:powerlaw} is small and may not have a significant qualitative effect on fiber dynamics in practice, even when multiple fibers are closely interacting. 

\vspace{1cm}

\subsection*{Acknowledgements}
L.O. is supported by NSF postdoctoral fellowship DMS-2001959.


\bibliographystyle{abbrv}
\bibliography{MRS_bib}

\end{document}